\documentclass{article}

\usepackage{PRIMEarxiv}

\usepackage[utf8]{inputenc} 
\usepackage[T1]{fontenc}    
\usepackage{hyperref}       
\usepackage{url}            
\usepackage{booktabs}       
\usepackage{amsfonts}       
\usepackage{nicefrac}       
\usepackage{microtype}      
\usepackage{lipsum}
\usepackage{fancyhdr}       
\usepackage{graphicx}       
\usepackage{tikz-cd}
\usepackage{amsmath,amssymb,amsthm,mathrsfs}
\theoremstyle{plain}
\newtheorem{theorem}{Theorem}
\newtheorem{lemma}[theorem]{Lemma}
\newtheorem{corollary}[theorem]{Corollary}
\newtheorem{example}[theorem]{Example}
\theoremstyle{definition}
\newtheorem{definition}[theorem]{Definition}
\newtheorem{remark}[theorem]{Remark}

\graphicspath{{media/}}     

\title{The two-sided Galois duals of multi-twisted codes
}

\author{
  Ramy F. Taki Eldin \\
  Faculty of Engineering \\
  Ain Shams University \\
  Cairo, Egypt\\
  \texttt{ramy.farouk@eng.asu.edu.eg} \\
}

\begin{document}
\maketitle

\begin{abstract}
Characterizing the duals of linear codes with rich algebraic structures received great interest in recent decades. The beginning was by representing cyclic codes over finite fields as ideals in the polynomial ring. Subsequently, studying the duals of constacyclic, quasi-cyclic, quasi-twisted, generalized quasi-cyclic, and multi-twisted codes appeared extensively in literature. We consider the class of multi-twisted (MT) codes because it extends to all of these codes. We describe a MT code $\mathcal{C}$ as a module over a principal ideal domain. Hence, $\mathcal{C}$ has a generator polynomial matrix (GPM) that satisfies an identical equation. The reduced GPM of $\mathcal{C}$ is the Hermite normal form of its GPM. We show that the Euclidean dual $\mathcal{C}^\perp$ of $\mathcal{C}$ is MT as well. We prove a formula for a GPM of $\mathcal{C}^\perp$ using the identical equation of the reduced GPM of $\mathcal{C}$. Then we aim to replace the Euclidean dual with the Galois dual. The Galois inner product is an asymmetric form, so we distinguish between the right and left Galois duals. We show that the right and left Galois duals of a MT code are MT as well but with possibly different shift constants. Our study is the first to contain the right and left Galois duals of a linear code simultaneously. This gives two advantages: establishing their interconnected identities and introducing the two-sided Galois dual that has not previously appeared in the literature. We use a condition for the two-sided Galois dual of a MT code to be MT, hence its GPM is characterized. Two special cases are also studied, one when the right and left Galois duals trivially intersect and the other when they coincide. The latter case is considered for any linear code, where a necessary and sufficient condition is established for the equality of the right and left Galois duals.
\end{abstract}

\keywords{Multi-twisted code  \and Generator polynomial matrix  \and Two-sided Galois dual  \and Hermite normal form}

%

\section{Introduction}
Linear codes with a rich algebraic structure are especially important in real communication systems due to the potential for developing their encoding and decoding algorithms. One of these codes is the class of cyclic codes over finite fields. Cyclic codes over the finite field $\mathbb{F}_q$ are in one-to-one correspondence with ideals of the polynomial ring $\mathbb{F}_q[x]$. The class of cyclic codes has undergone a series of generalizations to broader classes with other algebraic structures. The class of quasi-cyclic (QC) codes generalizes the shift index of cyclic codes so that it is not limited to a single cyclic shift, while constacyclic codes generalize the shift constant of cyclic codes to any nonzero element in the field. The algebraic structures of QC and constacyclic codes over finite fields are fully described in \cite{Barbier2012,SanLing2001,Cayrel2010} and \cite{Chen2012,Bakshi2012} respectively. In \cite{Esmaeili2009,Gneri2017,Siap05thestructure}, generalized quasi-cyclic (GQC) codes are presented as an in-depth generalization of QC codes. Unlike QC codes, the block lengths of GQC codes are not necessarily equal. On the other hand, \cite{Jia2012quasi,Gao2014} generalize the shift constant of QC codes to any nonzero element, and hence the class of quasi-twisted (QT) codes is obtained. However, the block lengths of a QT code are equal. In \cite{Aydin2017}, a comprehensive class of codes is introduced, the class of multi-twisted (MT) codes. This class contains cyclic codes, constacyclic codes, QC codes, QT codes, and GQC codes as subclasses. MT codes are similar to GQC codes in that the block lengths are not necessarily equal, and they are similar to QT codes in that the shift constants are not necessarily equal to one, moreover, different shift constants can be used for different blocks of a MT code. Algebraic structures for MT codes are described in \cite{Sharma2018,chauhan2021,Sharma}.

Due to the the invariance of QC, QT, and GQC codes under some linear transformations, they obtained the algebraic structure of modules over principal ideal domains \cite{Lally2001,Matsui2015}. Thus, any of these codes can be generated by a generator polynomial matrix (GPM) that satisfies some identical equation. GPM entries are elements in a principal ideal domain (PID), thus the Hermite normal form \cite{Cohen1996} is used to identify each code by a unique reduced matrix called the reduced GPM. In \cite{Matsui2015}, the identical equation of the reduced GPM is used to construct a GPM for the Euclidean dual of a GQC code. The Euclidean dual of a code is defined as the set of all vectors that yield a zero Euclidean inner product with each codeword in the code. In \cite{Fan2016}, the Euclidean inner product on the vector space $\mathbb{F}_q^n$ is generalized to the Galois inner product, and is used to determine the Galois duals of constacyclic codes. Later, the Galois duals of MT codes is discussed in \cite{Sharma}.

Our contribution to this paper is divided into several parts. In the first part, we generalize the description of GQC codes in \cite{Matsui2015} to describe MT codes as free modules over the PID $\mathbb{F}_q[x]$. We thus identify a MT code by the unique Hermite normal form of its GPM. Analogously, we set up an identical equation for the GPM. We take advantage of the proven fact that the Euclidean dual of a MT code is also MT and provide a GPM formula for this dual. Precisely, we use the identical equation of the MT code to deduce a GPM for its Euclidean dual. Although we imitate \cite{Matsui2015}, our result generalizes \cite{Matsui2015} to the inclusive class of MT codes. Specifically, we prove formulas for GPMs of the Euclidean duals of QC, QT, GQC, and MT codes from their identical equations.

In the second part, we aim to obtain a generalization of the first result by replacing the Euclidean inner product with the Galois inner product. We define the Galois inner product as in \cite{Fan2016} and demonstrate several interesting properties of Galois duals of MT codes. Unlike the Euclidean inner product, the Galois inner product is asymmetric. Consequently, the set of vectors that have zero Galois inner product with all codewords of a linear code will differ if we are going to make these vectors to the left or to the right of the Galois inner product. Therefore, it is necessary to differentiate between the right Galois dual and the left Galois dual of a linear code. To study these duals for MT codes, we inspect the application of a finite field automorphism to a MT code. This produces a MT code with the same block lengths but possibly different shift constants. We deduce the reduced GPM of the resulting MT code and the matrix that satisfies its identical equation from their counterparts of the original MT code. We then prove that the right and left Galois duals of a linear code are the images of its Euclidean dual under some automorphisms. Thus the right and left Galois duals of a MT code are also MT. We prove formulas for their shift constants, their reduced GPMs, and the matrices that satisfy their identical equations. To our knowledge, right and left Galois duals have not appeared simultaneously in any previous study. For instance, the Galois dual introduced in \cite{Hongwei} coincides with our definition of the right Galois dual, while the Galois dual introduced in \cite{Sharma} coincides with our definition of the left Galois dual. We find it useful to simultaneously include these two distinct Galois duals in our study. This allowed us to prove their interrelated properties, see Theorem \ref{Properties}. Some of these properties generalize the traditional properties of the Euclidean dual of a linear code. For example, the right (respectively, left) Galois dual of the left (respectively, right) Galois dual is the original code.

Another significant advantage of including the right and left Galois duals simultaneously in our study is to inspect the two-sided Galois dual of a MT code, which we define as the intersection of these two duals. For a MT code, although both its right and left Galois duals are MT, its two-sided Galois dual is not necessarily MT. We use a sufficient condition under which the two-sided Galois dual of a MT code is MT as well. Under this condition, we aim to describe a GPM of the two-sided Galois dual. We begin this direction in a more general context in Theorems \ref{Maiin1}--\ref{Maiin6}. A particular case of these theorems leads to some constraints whose solution produces the reduced GPM of the two-sided Galois dual and the matrix satisfying the identical equation. With the aid of the trace map over a finite field extension, we provide an auxiliary equation that helps in solving these constraints. An illustrative example shows in detail how to find a solution that satisfies these constraints. Furthermore, two remarkable cases of the two-sided Galois dual of a MT code are considered. The first is when the right and left Galois duals are identical. We establish a necessary and sufficient condition on any linear code to have equal right and left Galois duals. The second is when the right and left Galois duals trivially intersect. An application of the latter case is given in Corollary \ref{direct_sum}, which presents the condition on a MT code equivalent to writing the vector space as a direct sum of the right and left Galois duals of the code.

The remaining sections are organized as follows. Section \ref{preliminaries} summarizes some preliminaries to MT codes, their properties, GPMs, identical equations, and reduced forms of their GPMs. In Section \ref{Euclidean_dual}, we present our results regarding the Euclidean duals of MT codes. However, the results for the right, left, and two-sided Galois duals of a MT code are presented in Section \ref{Galois_dual_Sec}. We conclude the study in Section \ref{conclusion}.

\section{The algebraic structure of a MT code}
\label{preliminaries}
Let $\mathbb{F}_q$ be the finite field of order $q$, where $q=p^e$ is a prime power. A code $\mathcal{C}$ over $\mathbb{F}_q$ of length $n$ is linear if it is a subspace of $\mathbb{F}_q^n$, and hence we can define the dimension of $\mathcal{C}$. The Euclidean inner product on $\mathbb{F}_q^n$ is a symmetric bilinear form defined by
\[\langle\mathbf{a},\mathbf{b}\rangle =\sum_{i=0}^{n-1} a_i b_i\] 
for any $\mathbf{a}=\left(a_0,a_1,\ldots,a_{n-1}\right), \mathbf{b}=\left(b_0,b_1,\ldots,b_{n-1}\right)\in\mathbb{F}_q^n$. The Euclidean dual $\mathcal{C}^\perp$ of $\mathcal{C}$ is defined by
\begin{equation*}
\mathcal{C}^\perp=\left\{\mathbf{a}\in\mathbb{F}_q^n \ \mid \ \langle \mathbf{a},\mathbf{c}\rangle =0 \ \forall \ \mathbf{c}\in\mathcal{C} \right\}.
\end{equation*}
If $\mathcal{C}$ is linear of length $n$ and dimension $k$, one can easily show that $\mathcal{C}^\perp$ is linear of dimension $n-k$ and $\left(\mathcal{C}^\perp\right)^\perp=\mathcal{C}$.

A linear code is called cyclic if it is invariant under the cyclic shift of its codewords by one coordinate. That is, $\mathcal{C}$ is cyclic if and only if
\begin{equation*}
\left(c_0, c_1, \ldots, c_{n-2}, c_{n-1} \right)\in \mathcal{C} \Rightarrow \left(c_{n-1}, c_0, \ldots, c_{n-3}, c_{n-2} \right)\in \mathcal{C}.
\end{equation*}
It is convenient to represent the codewords of a cyclic code as polynomials in the quotient ring $\mathscr{R}=\mathbb{F}_q[x] / \langle x^n-1 \rangle$. Precisely, $\left(c_0, c_1, \ldots, c_{n-2}, c_{n-1} \right)\in\mathcal{C}$ has the polynomial representation $c_0 + c_1 x +\cdots + c_{n-2} x^{n-2} + c_{n-1} x^{n-1}\in \mathscr{R}$. This representation gives cyclic codes the structure of ideals in $\mathscr{R}$. The cyclic shift property of cyclic codes is generalized to constacyclic codes. Let $0\ne \lambda \in \mathbb{F}_q$. A linear code $\mathcal{C}$ is called constacyclic with a shift constant $\lambda$ if 
\begin{equation*}
\left(c_0, c_1, \ldots, c_{n-2}, c_{n-1} \right)\in \mathcal{C} \Rightarrow \left(\lambda c_{n-1}, c_0, \ldots, c_{n-3}, c_{n-2} \right)\in \mathcal{C}.
\end{equation*}
In polynomial representation, a constacyclic code over $\mathbb{F}_q$ of length $n$ and shift constant $\lambda$ is an ideal in the quotient ring $\mathscr{R}_\lambda=\mathbb{F}_q[x] / \langle x^n-\lambda \rangle$. But any ideal in $\mathscr{R}_\lambda$ corresponds to an ideal in $\mathbb{F}_q[x]$ containing $x^n-\lambda$. The latter has a unique monic generator polynomial $g(x)$ that satisfies the identical equation $a(x)g(x)=x^n-\lambda$; this is because $\mathbb{F}_q[x]$ is a PID. Thus, constacyclic codes over $\mathbb{F}_q$ of length $n$ and shift constant $\lambda$ are in one-to-one correspondence with ideals of $\mathbb{F}_q[x]$ generated by monic divisors of $x^n-\lambda$. We aim to present analogous correspondence in the class of MT codes.

A linear code $\mathcal{C}$ over $\mathbb{F}_q$ of length $n$ is called $\ell$-QC if it is invariant under the cyclic shift of its codewords by $\ell$ coordinates. Thus, $\mathcal{C}$ is $\ell$-QC if and only if
\begin{equation*}
\left( c_{1}, c_{2}, \ldots, c_{n}\right) \in\mathcal{C} \Rightarrow \left( c_{n-\ell+1}, \dots, c_{n}, c_{1}, c_{2},\ldots, c_{n-\ell}\right) \in\mathcal{C}.
\end{equation*}
The smallest positive integer $\ell$ with this property is called the index of $\mathcal{C}$ and denoted by $\ell$. Indeed, the index divides the code length $n$ and their quotient is called the co-index of $\mathcal{C}$, denoted $m$. A codeword of a QC code $\mathcal{C}$ of index $\ell$ and length $m\ell$ can be partitioned as 
\begin{equation}
\label{shift2}
\mathbf{c}=\left( c_{0,1}, c_{0,2}, \ldots, c_{0,\ell}, c_{1,1}, c_{1,2}, \ldots, c_{1,\ell}, \ldots, c_{m-1,1}, c_{m-1,2}, \ldots, c_{m-1,\ell} \right).
\end{equation}
A linear code $\mathcal{C}$ is QC of index $\ell$ and co-index $m$ if and only if
\begin{equation*}
\left( c_{m-1,1}, c_{m-1,2}, \ldots, c_{m-1,\ell}, c_{0,1}, c_{0,2}, \ldots, c_{0,\ell}, \ldots, c_{m-2,1}, c_{m-2,2}, \ldots, c_{m-2,\ell} \right)
\end{equation*}
is a codeword for every $\mathbf{c}\in\mathcal{C}$ in the form of \eqref{shift2}. QC codes generalize cyclic codes (when $\ell=1$) but not constacyclic codes, however QT codes do. For a nonzero $\lambda \in\mathbb{F}_q$, a linear code $\mathcal{C}$ over $\mathbb{F}_q$ of length $n$ is called $(\ell,\lambda)$-QT if 
\begin{equation*}
\left( c_{1}, c_{2}, \ldots, c_{n}\right) \in\mathcal{C} \Rightarrow \left( \lambda c_{n-\ell+1}, \dots, \lambda c_{n}, c_{1}, c_{2},\ldots, c_{n-\ell}\right) \in\mathcal{C}.
\end{equation*}
The index of $\mathcal{C}$ is the smallest positive integer $\ell$ with this property, while $\lambda$ is called the shift constant of $\mathcal{C}$. The index of a QT code divides its length and their quotient is the co-index $m$. Similar to QC codes, a linear code $\mathcal{C}$ of length $m\ell$ is $(\ell,\lambda)$-QT if and only if 
\begin{equation*}
\left( \lambda c_{m-1,1}, \lambda c_{m-1,2}, \ldots, \lambda c_{m-1,\ell}, c_{0,1}, c_{0,2}, \ldots, c_{0,\ell}, \ldots, c_{m-2,1}, c_{m-2,2},  \ldots, c_{m-2,\ell} \right)
\end{equation*}
is a codeword for every $\mathbf{c}\in\mathcal{C}$ in the form of \eqref{shift2}. Let $T_{(\ell,\lambda)}$ be the automorphism of $\mathbb{F}_q^{m\ell}$ such that
\begin{equation*}
\begin{split}
&T_{(\ell,\lambda)} \left( a_{0,1}, a_{0,2}, \ldots, a_{0,\ell}, a_{1,1}, a_{1,2}, \ldots, a_{1,\ell}, \ldots, a_{m-1,1}, a_{m-1,2}, \ldots, a_{m-1,\ell} \right)\\
&\quad  =  \left( \lambda a_{m-1,1}, \lambda a_{m-1,2}, \ldots , \lambda a_{m-1,\ell} , a_{0,1} , a_{0,2} , \ldots , a_{0,\ell} , \ldots , a_{m-2,1} , a_{m-2,2} , \ldots , a_{m-2,\ell} \right) .
\end{split}
\end{equation*}
We view $\mathbb{F}_q^{m\ell}$ as an $\mathbb{F}_q[x]$-module by defining the action of $x$ as the action of $T_{\ell,\lambda}$. Since an $(\ell,\lambda)$-QT code over $\mathbb{F}_q$ of length $m\ell$ is a $T_{\ell,\lambda}$-invariant $\mathbb{F}_q$-subspace of $\mathbb{F}_q^{m\ell}$, it is an $\mathbb{F}_q[x]$-submodule of $\mathbb{F}_q^{m\ell}$. To exhibit a polynomial representation for QT codes, let $\phi: \mathbb{F}_q^{m\ell} \rightarrow \mathscr{R}_\lambda^\ell$ be the $\mathbb{F}_q[x]$-module isomorphism defined by
\begin{eqnarray*}
\phi:\left( a_{0,1}, a_{0,2}, \ldots, a_{0,\ell}, a_{1,1}, a_{1,2}, \ldots, a_{1,\ell}, \ldots, a_{m-1,1}, a_{m-1,2}, \ldots, a_{m-1,\ell} \right) \mapsto \left( a_1(x), a_2(x), \ldots, a_\ell(x)\right)
\end{eqnarray*}
where $a_j(x)=a_{0,j}+a_{1,j} x+a_{2,j} x^2+\cdots +a_{m-1,j}x^{m-1}\in \mathscr{R}_\lambda$ for $1\le j \le \ell$. The polynomial representation of an $(\ell,\lambda)$-QT code $\mathcal{C}\subseteq \mathbb{F}_q^{m\ell}$ is $\phi\left(\mathcal{C}\right)$. Specifically, the codeword given by \eqref{shift2} is represented by the polynomial vector
\begin{equation*}
\begin{split}
\mathbf{c}(x)=&\left( c_{0,1}+c_{1,1} x+c_{2,1} x^2+\cdots +c_{m-1,1}x^{m-1},  c_{0,2}+c_{1,2} x+c_{2,2} x^2+\cdots +c_{m-1,2}x^{m-1}, \ldots,\right.\\ 
&\qquad\qquad\qquad\qquad\qquad\qquad\qquad\qquad\qquad\qquad\qquad \left. c_{0,\ell}+c_{1,\ell} x+c_{2,\ell} x^2+\cdots +c_{m-1,\ell}x^{m-1}\right)\in \mathscr{R}_\lambda^\ell.
\end{split}
\end{equation*}
Thus, $(\ell,\lambda)$-QT codes are in one-to-one correspondence with the $\mathbb{F}_q[x]$-submodules of $\mathscr{R}_\lambda^\ell$, and thus are in one-to-one correspondence with the $\mathbb{F}_q[x]$-submodules of $\left(\mathbb{F}_q[x]\right)^\ell$ containing the submodule
\begin{equation*}
M=\left( (x^m-\lambda)\mathbb{F}_q[x] \right)^\ell.
\end{equation*}
We do not distinguish between representing an $(\ell,\lambda)$-QT code as a $T_{\ell,\lambda}$-invariant subspace of $\mathbb{F}_q^{m\ell}$ or representing it as an $\mathbb{F}_q[x]$-submodule of $\left(\mathbb{F}_q[x]\right)^\ell$ that contains $M$. MT codes provide an additional generalization of QT codes by generalizing the $\ell$ block lengths of length $m$ into $\ell$ blocks that are not necessarily equal.
\begin{definition}
Let $m_1, m_2, \ldots ,m_\ell$ be positive integers and $\Lambda=\left(\lambda_1, \lambda_2, \ldots, \lambda_\ell \right)$, where $0\ne\lambda_j \in\mathbb{F}_q$ for $1\le j\le \ell$. A $\Lambda$-MT code over $\mathbb{F}_q$ of index $\ell$ and block lengths $(m_1,m_2,\ldots,m_\ell)$ is an $\mathbb{F}_q[x]$-submodule of $\left(\mathbb{F}_q[x]\right)^\ell$ that contains the submodule
\begin{equation*}
\begin{split}
M_\Lambda= \bigoplus_{j=1}^{\ell} \left(\left(x^{m_j}-\lambda_j\right)\mathbb{F}_q[x]\right).
\end{split}
\end{equation*}\end{definition}
From its definition, a MT code $\mathcal{C}$ of index $\ell$ is a linear code over $\mathbb{F}_q[x]$ of length $\ell$. A generator matrix for $\mathcal{C}$ as a linear code over $\mathbb{F}_q[x]$ is called GPM because its entries are polynomials over $\mathbb{F}_q$. Since a GPM is a matrix over the PID $\mathbb{F}_q[x]$, one might ask for its unique Hermite normal form, which we call the reduced GPM.

\begin{theorem}
\label{MT_Corresp}
There is a one-to-one correspondence between $\left(\lambda_1, \lambda_2, \ldots, \lambda_\ell \right)$-MT codes over $\mathbb{F}_q$ of index $\ell$ and block lengths $(m_1,m_2,\ldots,m_\ell)$ and $T_{\Lambda}$-invariant $\mathbb{F}_q$-subspaces of $\mathbb{F}_q^{n}$, where $n=m_1+m_2+\cdots+m_\ell$ and $T_{\Lambda}$ is the automorphism of $\mathbb{F}_q^n$ given by
\begin{equation}
\label{T_ell_Lambda}
\begin{split}
&T_{\Lambda}:\left( a_{0,1},\ldots, a_{m_1-1,1},a_{0,2},\ldots, a_{m_2-1,2},\ldots, a_{0,\ell},\ldots, a_{m_\ell-1,\ell}\right)\mapsto\\
&\qquad\left(\lambda_1 a_{m_1-1,1},a_{0,1},\ldots, a_{m_1-2,1}, \lambda_2 a_{m_2-1,2},a_{0,2},\ldots, a_{m_2-2,2},\ldots, \lambda_\ell a_{m_\ell-1,\ell},a_{0,\ell},\ldots, a_{m_\ell-2,\ell}\right).
\end{split}
\end{equation}
\end{theorem}
\begin{proof}
For $1\le j\le \ell$, let $\mathscr{R}_{m_j,\lambda_j}=\mathbb{F}_q[x]/\langle x^{m_j}-\lambda_j\rangle$ and $\pi_j:\mathbb{F}_q[x]\rightarrow\mathscr{R}_{m_j,\lambda_j}$ be the projection homomorphism. Then $\pi=\oplus_{j=1}^\ell \pi_j : \left(\mathbb{F}_q[x]\right)^\ell \rightarrow \oplus_{j=1}^\ell \mathscr{R}_{m_j,\lambda_j}$ is a surjective homomorphism with kernel $M_\Lambda$. Actually, $\pi$ defines a one-to-one correspondence between $\mathbb{F}_q[x]$-submodules of $\oplus_{j=1}^\ell \mathscr{R}_{m_j,\lambda_j}$ and $\mathbb{F}_q[x]$-submodules of $\left(\mathbb{F}_q[x]\right)^\ell$ that contain $M_\Lambda$. Hence, $\left(\lambda_1, \lambda_2, \ldots, \lambda_\ell \right)$-MT codes over $\mathbb{F}_q$ of index $\ell$ and block lengths $(m_1,m_2,\ldots,m_\ell)$ are precisely the submodules of $\oplus_{j=1}^\ell \mathscr{R}_{m_j,\lambda_j}$.

We view $\mathbb{F}_q^{n}$ as an $\mathbb{F}_q[x]$-module by defining the action of $x$ as the action of $T_{\Lambda}$. Then the $\mathbb{F}_q[x]$-submodules of $\mathbb{F}_q^{n}$ are precisely the $T_{\Lambda}$-invariant $\mathbb{F}_q$-subspaces of $\mathbb{F}_q^{n}$. Let $\phi: \mathbb{F}_q^n\rightarrow \oplus_{j=1}^\ell \mathscr{R}_{m_j,\lambda_j}$ be the $\mathbb{F}_q$-vector space isomorphism defined by
\begin{equation}
\label{the_isomorphism_phi}
\left( a_{0,1},\ldots, a_{m_1-1,1},\ldots, a_{0,\ell},\ldots, a_{m_\ell-1,\ell}\right) \mapsto \left( a_1(x),a_2(x),\ldots,a_\ell(x)\right)
\end{equation}
where $a_j(x)=a_{0,j}+a_{1,j}x+\cdots+ a_{m_j-1,j}x^{m_j-1}$ for $1\le j\le \ell$. This gives the commutative diagram of $\mathbb{F}_q$-vector space isomorphisms
\begin{equation}
\label{Commut_diagram2}
\begin{tikzcd}[swap]
    \mathbb{F}_q^{n} \arrow{r}[swap]{\phi}{} \arrow{d}{T_{\Lambda}}  & \oplus_{j=1}^\ell \mathscr{R}_{m_j,\lambda_j} \arrow{d}[swap]{\psi}{} \\  
    \mathbb{F}_q^{n} \arrow{r}[swap]{}{\phi}   & \oplus_{j=1}^\ell \mathscr{R}_{m_j,\lambda_j}
\end{tikzcd}
\end{equation}
where $\psi:  \left( a_1(x), a_2(x), \ldots, a_\ell(x)\right) \mapsto \left(x a_1(x),x a_2(x), \ldots,x a_\ell(x)\right)$. Then $x\phi\left(\mathbf{a}\right)=\phi\left(T_{\Lambda}(\mathbf{a})\right)$ for any $\mathbf{a}\in\mathbb{F}_q^{n}$, and $\phi$ is an $\mathbb{F}_q[x]$-module isomorphism.

If $\mathcal{C}$ is a $\left(\lambda_1, \lambda_2, \ldots, \lambda_\ell \right)$-MT code over $\mathbb{F}_q$ of index $\ell$ and block lengths $(m_1,m_2,\ldots,m_\ell)$, then $\phi^{-1}\circ\pi\left(\mathcal{C}\right)$ is a $T_{\Lambda}$-invariant $\mathbb{F}_q$-subspace of $\mathbb{F}_q^{n}$. Conversely, the image of any $T_{\Lambda}$-invariant $\mathbb{F}_q$-subspace of $\mathbb{F}_q^{n}$ under the map $\pi^{-1}\circ\phi$ is an $\mathbb{F}_q[x]$-submodule of $\left(\mathbb{F}_q[x]\right)^\ell$ that contains $M_\Lambda$.
\end{proof}

Hereinafter, by a MT code we mean a $T_{\Lambda}$-invariant subspace of $\mathbb{F}_q^{n}$ or a submodule of $\left(\mathbb{F}_q[x]\right)^\ell$ that contains $M_\Lambda$, and the used algebraic structure is determined from the context. On the other hand, the polynomial representation of a MT-code is the corresponding submodule of $\oplus_{j=1}^\ell \mathscr{R}_{m_j,\lambda_j}$.

Let $\mathcal{C}$ be a $\left(\lambda_1, \lambda_2, \ldots, \lambda_\ell \right)$-MT code over $\mathbb{F}_q$ of index $\ell$, block lengths $(m_1,m_2,\ldots,m_\ell)$, and an $r\times n$ generator matrix $G$ that generates $\mathcal{C}$ as an $\mathbb{F}_q$-subspace of $\mathbb{F}_q^{n}$. Let $\phi$ be the map defined by \eqref{the_isomorphism_phi} and let 
\begin{equation*}
\left( G_{i,1}(x) , G_{i,2}(x) ,\ldots, G_{i,\ell}(x)\right)=\phi\left( \mathrm{row}_i\left(G\right)\right)
\end{equation*}
for $i=1,2,\ldots, r$. Then $\mathcal{C}$ (as an $\mathbb{F}_q[x]$-submodule of $\left(\mathbb{F}_q[x]\right)^\ell$) has a GPM of the form
\begin{equation*}
\begin{pmatrix}
G_{1,1}(x) & G_{1,2}(x) & G_{1,3}(x) & \cdots & G_{1,\ell}(x)\\
G_{2,1}(x) & G_{2,2}(x) & G_{2,3}(x) & \cdots & G_{2,\ell}(x)\\
\vdots & \vdots & \vdots & \ddots & \vdots \\
G_{r,1}(x) & G_{r,2}(x) & G_{r,3}(x) & \cdots & G_{r,\ell}(x)\\
x^{m_1}-\lambda_1 & 0 & 0 & \cdots & 0\\
0 & x^{m_2}-\lambda_2 & 0 &  \cdots & 0\\
\vdots & \vdots & \vdots & \ddots & \vdots \\
0 & 0 & 0 & \cdots & x^{m_\ell}-\lambda_\ell
\end{pmatrix}.
\end{equation*}
Reducing this matrix to the Hermite normal form yields the reduced GPM $\mathbf{G}$ of $\mathcal{C}$. In fact, $\left(\mathbb{F}_q[x]\right)^\ell$ and $M_\Lambda$ are free modules of rank $\ell$ over the PID $\mathbb{F}_q[x]$ and $M_\Lambda\subseteq \mathcal{C}\subseteq \left(\mathbb{F}_q[x]\right)^\ell$, then $\mathcal{C}$ has rank $\ell$. Consequently, the reduced GPM $\mathbf{G}=\left[g_{i,j}\right]$ is upper triangular of rank $\ell$ and size $\ell\times\ell$ such that, for $1\le i\le \ell$,
\begin{enumerate}
\item $g_{i,i}\ne 0$ is monic and
\item $\deg\left(g_{h,i}\right)<\deg\left(g_{i,i}\right)$ for all $1\le h<i$.
\end{enumerate}

\begin{theorem}
\label{Containment}
Let $\mathcal{C}$ and $\mathcal{C}'$ be two $\Lambda$-MT codes of index $\ell$ and block lengths $\left(m_1,m_2,\ldots,m_\ell\right)$. Let $\mathbf{G}$ and $\mathbf{G}'$ be GPMs for $\mathcal{C}$ and $\mathcal{C}'$ respectively. Then, $\mathcal{C}'\subseteq\mathcal{C}$ if and only if $\mathbf{G}'=\mathbf{Y}\mathbf{G}$ for some matrix $\mathbf{Y}$. If $\mathbf{G}$ and $\mathbf{G}'$ are the reduced GPMs, then $\mathbf{Y}$ is upper triangular.
\end{theorem}
\begin{proof}
{We have $\mathcal{C}'\subseteq\mathcal{C}$ if and only if $\mathbf{G}$ generates the rows of $\mathbf{G}'$ if and only if $\mathbf{G}'=\mathbf{Y}\mathbf{G}$ for some matrix $\mathbf{Y}$. Suppose $\mathbf{G}$ and $\mathbf{G}'$ are in the reduced form. Then $\mathbf{Y}$ is upper triangular because $\mathbf{G}$ and $\mathbf{G}'$ are upper triangular with nonzero diagonal entries and $\mathbf{F}_q[x]$ is an integral domain.}
\end{proof}

The diagonal matrix 
\begin{equation*}
\mathbf{D}=\mathrm{diag}\left[x^{m_1}-\lambda_1,x^{m_2}-\lambda_2,\ldots,x^{m_\ell}-\lambda_\ell\right]
=\begin{pmatrix}
x^{m_1}-\lambda_1 & 0 & \ldots & 0\\
0 & x^{m_2}-\lambda_2 & \ldots & 0\\
\vdots & \vdots & \ddots & \vdots \\
0 & 0 & \ldots & x^{m_\ell}-\lambda_\ell
\end{pmatrix}
\end{equation*}
is the reduced GPM of the $\Lambda$-MT code $M_\Lambda$. But any $\Lambda$-MT code $\mathcal{C}$ with a GPM $\mathbf{G}$ contains $M_\Lambda$. Then from Theorem \ref{Containment}, there is a matrix $\mathbf{A}$ such that
\begin{equation}
\label{identical_eq}
\mathbf{A}\mathbf{G}=\mathbf{D}.
\end{equation}
Equation \eqref{identical_eq} is called the identical equation of $\mathbf{G}$. The matrix $\mathbf{A}$ plays a fundamental role in constructing a GPM for the Euclidean and Galois duals of a MT code. If $\mathbf{A}=\left[a_{i,j}\right]$ is the matrix that satisfies the identical equation of the reduced GPM, then $\mathbf{A}$ is upper triangular and for $1\le i\le \ell$
\begin{enumerate}
\item $a_{i,i}=\frac{x^{m_i}-\lambda_i}{g_{i,i}}$ and
\item $\deg\left(a_{i,h}\right)<\deg\left(a_{i,i}\right)$ for all $i< h\le \ell$.
\end{enumerate}

In particular, $\mathbf{A}$ and $\mathbf{G}$ commute when $\mathcal{C}$ is $(\ell,\lambda)$-QT.
\begin{theorem}
\label{A_G_Commutativity}
Let $\mathcal{C}$ be an $(\ell,\lambda)$-QT code with a GPM $\mathbf{G}$ and let $\mathbf{A}$ be the matrix that satisfies the identical equation of $\mathbf{G}$. Then $\mathbf{G}\mathbf{A}=\mathbf{D}$.\end{theorem}
\begin{proof}
{Assume $\mathbf{G}\mathbf{A}=\mathbf{B}$. Then $\mathbf{B}\mathbf{G}=\mathbf{G}\mathbf{A}\mathbf{G}=\mathbf{G}\mathbf{D}=\left(x^m-\lambda\right)\mathbf{G}=\mathbf{D}\mathbf{G}$. That is, $\left(\mathbf{B}-\mathbf{D}\right)\mathbf{G}=\mathbf{0}$. Rows of $\mathbf{G}$ form a basis for $\mathcal{C}$, then $\left(\mathbf{B}-\mathbf{D}\right)=\mathbf{0}$ and $\mathbf{B}=\mathbf{D}$.}
\end{proof}

The following result can be proven in a similar way to Corollary 3.1 in \cite{Sharma}.
\begin{theorem}
\label{dim_MT}
Let $\mathcal{C}$ be a $\Lambda$-MT code over $\mathbb{F}_q$ of index $\ell$ and block lengths $\left(m_1,m_2,\ldots,m_\ell\right)$ and let $\mathbf{G}=[g_{i,j}]$ be an upper triangular GPM of $\mathcal{C}$. Then $\mathcal{C}$ has dimension
\begin{equation*}
k=\sum_{j=1}^{\ell}\left( m_j-\mathrm{deg}(g_{j,j})\right).
\end{equation*}
as an $\mathbb{F}_q$-vector space. Equivalently, $k=\deg\left(\mathrm{det}\left(\mathbf{A}\right)\right)$, where $\mathbf{A}$ is the matrix that satisfies the identical equation of $\mathbf{G}$ and $\mathrm{det}\left(\mathbf{A}\right)$ is the determinant of $\mathbf{A}$.
\end{theorem}

\begin{example}
\label{ex_Sh2}
Let $\mathcal{C}$ be the $(2,1)$-MT code over $\mathbb{F}_3$ of index $\ell=2$, block lengths $\left(m_1,m_2\right)=\left(20,40\right)$, and the reduced GPM
\begin{equation*}
\mathbf{G}=\begin{pmatrix}
g_{1,1} & g_{1,2}\\
0 & x^{40}+2
\end{pmatrix}
\end{equation*}
where $g_{1,1}= 2+ x+2 x^2+ x^3+ x^4+2 x^5+ x^7+ x^9+2 x^{10}+ x^{11}+2 x^{13}+x^{14}$ and $g_{1,2}= x+ x^4+ x^5+ x^7+2 x^9+2 x^{11}+2 x^{12}+ x^{13}+ x^{14}+ x^{16}+ x^{17}+2 x^{19}+2 x^{21}+2 x^{24}+2 x^{25}+2 x^{27}+ x^{29}+ x^{31}+ x^{32}+2 x^{33}+2 x^{34}+2 x^{36}+2 x^{37}+x^{39}$. The matrix that satisfies the identical equation of $\mathbf{G}$ is
\begin{equation*}
\mathbf{A}=\begin{pmatrix}
2+2 x + x^4 + x^5+x^6 & \quad 2x(1+x)^4\\
0 & 1
\end{pmatrix}.
\end{equation*}
From Theorem \ref{dim_MT}, $\mathcal{C}$ has dimension $k=(20-14)+(40-40)=6$.
\end{example}

\begin{definition}
An $\ell$-GQC code over $\mathbb{F}_q$ of block lengths $(m_1,m_2,\ldots,m_\ell)$ is a MT code of index $\ell$, block lengths $\left(m_1,m_2,\ldots,m_\ell\right)$, and shift constants $\lambda_j=1$ for $1\le j\le \ell$.
\end{definition}
From Theorem \ref{MT_Corresp}, an $\ell$-GQC code can be thought of as: 
\begin{enumerate}
\item An $\mathbb{F}_q[x]$-submodule of $\left(\mathbb{F}_q[x]\right)^\ell$ that contains $\oplus_{j=1}^{\ell} \left(\left(x^{m_j}-1\right)\mathbb{F}_q[x]\right)$.
\item An invariant $\mathbb{F}_q$-subspace of $\mathbb{F}_q^{n}$, where $n=\sum_{j=1}^\ell m_j$, under the automorphism
\begin{equation*}
\begin{split}
&\left( a_{0,1},\ldots, a_{m_1-1,1},a_{0,2},\ldots, a_{m_2-1,2},\ldots, a_{0,\ell},\ldots, a_{m_\ell-1,\ell}\right)\mapsto\\
&\left( a_{m_1-1,1},  a_{0,1},  \ldots,   a_{m_1-2,1},   a_{m_2-1,2},  a_{0,2},  \ldots ,  a_{m_2-2,2},  \ldots, a_{m_\ell-1,\ell},  a_{0,\ell},  \ldots,  a_{m_\ell-2,\ell} \right).
\end{split}
\end{equation*} 
\item An $\mathbb{F}_q[x]$-submodule of $\oplus_{j=1}^\ell \mathscr{R}_{m_j,1}$, where $\mathscr{R}_{m_j,1}=\mathbb{F}_q[x]/\langle x^{m_j}-1\rangle$.
\end{enumerate}

\section{Euclidean Duals of MT codes}
\label{Euclidean_dual}
In this section, we focus on discussing the Euclidean duals of MT codes. Unless otherwise stated in this section, let $\mathcal{C}$ denote a $\Lambda$-MT code over $\mathbb{F}_q$ of index $\ell$ and block lengths $\left(m_1,m_2,\ldots,m_\ell\right)$, where $\Lambda=\left(\lambda_1,\lambda_2,\ldots,\lambda_\ell \right)$ while $0\ne\lambda_j\in\mathbb{F}_q$ and $m_i$ is a positive integer for $1\le j\le\ell$. We also let $\mathbf{G}$ be a GPM for $\mathcal{C}$, and we denote the matrix that satisfies the identical equation of $\mathbf{G}$ by $\mathbf{A}$. In the following result, we prove that the Euclidean dual $\mathcal{C}^\perp$ of $\mathcal{C}$ is not only linear, but also MT with the same block lengths but possibly different shift constants. However, the main result of this section is to derive a formula for a GPM of $\mathcal{C}^\perp$. This will be achieved with the aid of the identical equation of $\mathbf{G}$.

\begin{theorem}
\label{MT_dual}
The Euclidean dual $\mathcal{C}^\perp$ of $\mathcal{C}$ is $\Delta$-MT of block lengths $\left(m_1,m_2,\ldots,m_\ell\right)$, where $\Delta=\left(\frac{1}{\lambda_1},\frac{1}{\lambda_2},\ldots,\frac{1}{\lambda_\ell} \right)$.\end{theorem}
\begin{proof}
{From Theorem \ref{MT_Corresp}, $\mathcal{C}$ is a $T_{\Lambda}$-invariant subspace of $\mathbb{F}_q^{n}$, where $n=\sum_{j=1}^\ell m_j$ and $T_{\Lambda}$ is the automorphism given by \eqref{T_ell_Lambda}. Let $N=\mathrm{lcm}\left(t_1 m_1,t_2 m_2,\ldots, t_\ell m_\ell \right)$, where $t_j$ is the multiplicative order of $\lambda_j$ for $1\le j\le \ell$. Observe that applying $T_{\Lambda}$ exactly $N$ times to any $\mathbf{a}\in\mathbb{F}_q^{n}$ keeps $\mathbf{a}$ unchanged. Thus $T_{\Lambda}^N$ is the identity map on $\mathbb{F}_q^{n}$. If we can show that $T_{\Delta}\left(\mathcal{C}^\perp\right)=\mathcal{C}^\perp$, then $\mathcal{C}^\perp$ is $\Delta$-MT. To do this, consider any $\mathbf{b}\in\mathcal{C}^\perp$ and $\mathbf{c}\in\mathcal{C}$. Then
\begin{equation*}
\langle\mathbf{c}, T_{\Delta}\left(\mathbf{b}\right)\rangle
=\langle T_{\Lambda}^N\left(\mathbf{c}\right), T_{\Delta}\left(\mathbf{b}\right)\rangle
=\langle T_{\Lambda}\circ T_{\Lambda}^{N-1}\left(\mathbf{c}\right), T_{\Delta}\left(\mathbf{b}\right)\rangle
=\langle T_{\Lambda}^{N-1}\left(\mathbf{c}\right), \mathbf{b}\rangle
=0
\end{equation*}
because $T_{\Lambda}^{N-1}\left(\mathbf{c}\right)\in\mathcal{C}$. Then, $T_{\Delta}\left(\mathbf{b}\right)\in\mathcal{C}^\perp$ and $T_{\Delta}\left(\mathcal{C}^\perp\right)\subseteq\mathcal{C}^\perp$. Equality holds since $T_{\Delta}$ is a vector space automorphism.}
\end{proof}

Now we define some matrices that are jointly related to the matrix that satisfies the identical equation of the reduced GPM of $\mathcal{C}$.
\begin{definition}
\label{def_parity}
For a MT code $\mathcal{C}$, let $\mathbf{G}=[g_{i,j}]$ be the reduced GPM of $\mathcal{C}$ and let $\mathbf{A}=[a_{i,j}]$ be the matrix that satisfies the identical equation of $\mathbf{G}$. For $1\le j\le \ell$, denote the degree of $\deg\left(g_{j,j}\right)$ by $d_j$, i.e., $d_j=\deg\left(g_{j,j}\right)$ . 
\begin{enumerate}
\item Let $\mathbf{A}\left( \frac{1}{x}\right)$ be the matrix obtained from $\mathbf{A}$ when $x$ is replaced by $\frac{1}{x}$.
\item Let $\mathbf{A}^*$ be the matrix obtained after multiplying the $(i,j)$-th entry of $\mathbf{A}\left( \frac{1}{x}\right)$ by $x^{m_i-d_j}$.
\item (Eliminate the negative exponents in $\mathbf{A}^*$) Let $\mathbf{A}^{**}$ be the matrix obtained from $\mathbf{A}^*$ by reducing the $(i,j)$-th entry (for $i<j$) of $\mathbf{A}^*$ modulo $\left(x^{m_i}-\frac{1}{\lambda_i}\right)$. Specifically, $x^{-\mu}$ is replaced by $\lambda_i x^{m_i-\mu}$ for $\mu\ge 1$.
\item Let $\mathbf{H}=\left(\mathbf{A}^{**}\right)^t$,  where $^t$ stands for matrix transpose.
\end{enumerate}
\end{definition}

From Theorem \ref{MT_dual}, $\mathcal{C}^\perp$ is $\Delta$-MT code, where $\Delta=\left(\frac{1}{\lambda_1},\frac{1}{\lambda_2},\ldots,\frac{1}{\lambda_\ell} \right)$. For $1\le h\le \ell$, let 
\begin{equation*}
\pi_h:\mathbb{F}_q[x]\rightarrow\mathscr{R}_{m_h,\frac{1}{\lambda_h}}=\mathbb{F}_q[x]/\langle x^{m_h}-\lambda_h^{-1}\rangle
\end{equation*}
be the projection homomorphism and let $\pi=\oplus_{h=1}^\ell \pi_h$. View $\mathbb{F}_q^{n}$ as an $\mathbb{F}_q[x]$-module by defining the action of $x$ as the action of $T_{\Delta}$. Define the $\mathbb{F}_q[x]$-module isomorphism $\phi$ by
\begin{gather*}
\mathbb{F}_q^n\rightarrow \oplus_{h=1}^\ell \mathscr{R}_{m_h,\frac{1}{\lambda_h}}\\
\left( b_{0,1},\ldots, b_{m_1-1,1},\ldots, b_{0,\ell},\ldots, b_{m_\ell-1,\ell}\right) \mapsto \left( b_1(x),b_2(x),\ldots,b_\ell(x)\right)
\end{gather*}
where $b_h (x)=b_{0,h}+b_{1,h}x+\cdots+ b_{m_h-1,h}x^{m_h-1}$ for $1\le h\le \ell$. Similar to \eqref{Commut_diagram2}, we construct the commutative diagram
\begin{equation*}
\begin{tikzcd}[swap]
    \mathbb{F}_q^{n} \arrow{r}[swap]{\phi}{} \arrow{d}{T_{\Delta}}  & \oplus_{h=1}^\ell \mathscr{R}_{m_h,\frac{1}{\lambda_h}} \arrow{d}[swap]{\psi}{} & \left(\mathbb{F}_q[x]\right)^\ell \arrow{l}[swap]{}{\pi} \arrow{d}[swap]{}{\psi}\\  
    \mathbb{F}_q^{n} \arrow{r}[swap]{}{\phi}   & \oplus_{h=1}^\ell \mathscr{R}_{m_h,\frac{1}{\lambda_h}} & \left(\mathbb{F}_q[x]\right)^\ell \arrow{l}[swap]{\pi}{}
\end{tikzcd}
\end{equation*}
where $\psi:  \left( b_1(x), b_2(x), \ldots, b_\ell(x)\right) \mapsto \left(x b_1(x),x b_2(x), \ldots,x b_\ell(x)\right)$.

Let us fix a positive integer $j\le \ell$ and argue as in Definition \ref{def_parity}. Suppose that the $j^{\mathrm{th}}$ column of $\mathbf{A}$ is
\begin{equation*}
\left(a_{1,j}(x), a_{2,j}(x), \ldots, a_{\ell,j}(x) \right)^t
\end{equation*}
where $\deg\left(a_{h,j}(x)\right) < \deg\left(a_{h,h}(x)\right) \le m_h$ for $h< j$, $a_{h,j}(x)=0$ for $h>j$, and $\deg\left(a_{j,j}(x)\right) < m_j$ or $a_{j,j}(x)=x^{m_j}-\lambda_j$. 
Then the $j^{\mathrm{th}}$ column of $\mathbf{A}^*$ is 
\begin{equation*}
\left(x^{m_1-d_j}a_{1,j}\left(\frac{1}{x}\right), x^{m_2-d_j}a_{2,j}\left(\frac{1}{x}\right), \ldots, x^{m_\ell-d_j}a_{\ell,j}\left(\frac{1}{x}\right) \right)^t.
\end{equation*}
The $j^{\mathrm{th}}$ row of $\mathbf{H}$ is the $j^{\mathrm{th}}$ column of $\mathbf{A}^{**}$ and it satisfies  
\begin{equation*}
\pi\left( \mathrm{row}_j\left(\mathbf{H}\right)\right)=\pi\left( \left(\mathrm{column}_j\left(\mathbf{A}^*\right)\right)^t\right).
\end{equation*}
Let $\mathbf{a}_j=T_{\Delta}^{d_j-1}\left(\phi^{-1}\left( \pi\left( \mathrm{row}_j\left( \mathbf{H}\right)\right)\right)\right)$. Then
\begin{equation}
\label{word_a}
\mathbf{a}_j=\left(a_{m_1-1,1,j},\ldots,a_{0,1,j},a_{m_2-1,2,j},\ldots,a_{0,2,j},\ldots,a_{m_\ell-1,\ell,j},\ldots,a_{0,\ell,j} \right)
\end{equation}
where $\pi_h\left(a_{h,j}(x)\right)=a_{0,h,j}+a_{1,h,j}x +\cdots+a_{m_h-2,h,j}x^{m_h-2}+a_{m_h-1,h,j}x^{m_h-1}$ for $1\le h\le \ell$. Since $\mathcal{C}^\perp$ is $T_{\Delta}$-invariant, $\mathrm{row}_j\left( \mathbf{H}\right)\in\mathcal{C}^\perp$ if and only if $\mathbf{a}_j$ gives zero inner product with each codeword in $\mathcal{C}$ and that is actually what we will show in the next result.

\begin{lemma}
\label{lemma1}
For any positive integer $j\le \ell$, $\mathrm{row}_j\left( \mathbf{H}\right)\in\mathcal{C}^\perp$.
\end{lemma}
\begin{proof}
 For $1\le h\le \ell$, let $t_h$ be the multiplicative order of $\lambda_h$, let $N=\mathrm{lcm}\left(t_1m_1,t_2m_2,\ldots,t_\ell m_\ell\right)$, and let $N_h=N/( m_h t_h)$. Then
\begin{align*}
\frac{x^N-1}{ x^{m_h}-\lambda_h}=\frac{x^{m_h t_h N_h}-\lambda_h^{t_h N_h}}{ x^{m_h}-\lambda_h}=x^{N-m_h}+\lambda_h x^{N-2m_h}+\lambda_h^2 x^{N-3m_h}+\cdots +\lambda_h^{t_h N_h-2}x^{m_h}+\lambda_h^{t_h N_h-1}.
\end{align*}
Let $\mathbf{T}=\mathrm{diag}\left[\frac{x^N-1}{x^{m_1}-\lambda_1}, \ldots, \frac{x^N-1}{x^{m_\ell}-\lambda_\ell}\right]$. From \eqref{identical_eq}, 
\begin{equation*}
\mathbf{A}\left(\mathbf{G}\mathbf{T}\right)=\mathbf{A}\mathbf{G}\mathbf{T}=\mathbf{D}\mathbf{T}=(x^N-1)\mathbf{I}_\ell.
\end{equation*}
Applying the same argument in the proof of Theorem \ref{A_G_Commutativity}, we get 
\begin{equation}
\label{In_proof_Lemma1}
\left(\mathbf{G}\mathbf{T} \right)\mathbf{A}=(x^N-1)\mathbf{I}_\ell.
\end{equation}
Reducing the $(i,j)$-th entry of \eqref{In_proof_Lemma1} modulo $\left(x^N-1\right)$ leads to
\begin{equation}
\label{inproof_lem1}
\begin{split}
&\sum_{h=1}^{\ell} g_{i,h}\left( \frac{x^N-1}{x^{m_h}-\lambda_h} \right) a_{h,j}\\
&=\sum_{h=1}^{\ell}  \left( x^{N-m_h}+\lambda_h x^{N-2m_h}+\lambda_h^2 x^{N-3m_h}+\cdots +\lambda_h^{t_h N_h-2}x^{m_h}+\lambda_h^{t_h N_h-1} \right) g_{i,h}a_{h,j}\\
&\equiv 0 \pmod{x^N-1}.
\end{split}
\end{equation} 
In fact, $\mathbf{G}$ is the reduced GPM of $\mathcal{C}$. Then $g_{i,h}=0$ for $h<i$, $\deg\left(g_{i,h} \right)<m_h$ for any $h>i$, and $\deg\left(g_{i,i} \right)<m_i$ or $g_{i,i}=x^{m_i}-\lambda_i$. If $g_{i,i}=x^{m_i}-\lambda_i$, then $g_{i,i}\left( \frac{x^N-1}{x^{m_i}-\lambda_i} \right) a_{i,j}\equiv 0\pmod{x^N-1}$. Thus, in all cases, $g_{i,h}$ in \eqref{inproof_lem1} can be replaced by 
\begin{equation*}
\pi_h\left( g_{i,h}\right)=g_{0,i,h}+g_{1,i,h}x+\cdots +g_{m_h-2,i,h}x^{m_h-2}+g_{m_h-1,i,h}x^{m_h-1}. 
\end{equation*}
Similarly, $a_{h,j}=0$ for $h>j$, $\deg\left(a_{h,j} \right)<m_h$ for $h<j$, and $\deg\left(a_{j,j} \right)<m_j$ or $a_{j,j}=x^{m_j}-\lambda_j$. If $a_{j,j}=x^{m_j}-\lambda_j$, then $g_{i,j}\left( \frac{x^N-1}{x^{m_j}-\lambda_j} \right) a_{j,j}\equiv 0 \pmod{x^N-1}$. Thus, in all cases, $a_{h,j}$ in \eqref{inproof_lem1} can be replaced by 
\begin{equation*}
\pi_h\left( a_{h,j}\right)=a_{0,h,j}+a_{1,h,j}x+\cdots +a_{m_h-2,h,j}x^{m_h-2}+a_{m_h-1,h,j}x^{m_h-1}. 
\end{equation*}
Then
\begin{equation}
\label{inproof_lem1_2}
\begin{split}
&\sum_{h=1}^{\ell} \left(x^{N-m_h}+\lambda_h x^{N-2m_h}+\lambda_h^2 x^{N-3m_h}+\cdots +\lambda_h^{t_h N_h-2}x^{m_h}+\lambda_h^{t_h N_h-1}\right)\\
&\qquad \left( g_{0,i,h}+g_{1,i,h}x+\cdots +g_{m_h-2,i,h}x^{m_h-2}+g_{m_h-1,i,h}x^{m_h-1}\right)\\
&\qquad\left( a_{0,h,j}+a_{1,h,j}x+\cdots +a_{m_h-2,h,j}x^{m_h-2}+a_{m_h-1,h,j}x^{m_h-1}\right)\\
&\equiv 0 \pmod{x^N-1}.
\end{split}
\end{equation} 
For any integer $0\le \nu\le N-1$, the sum of the coefficients of $x^{N-\nu-1}$ and $x^{2N-\nu-1}$ in \eqref{inproof_lem1_2} is zero. What this shows is that the inner product of $\mathbf{a}_j$ (see Equation \eqref{word_a}) and $T_{\Lambda}^{\nu}\left(\phi^{-1}\left( \pi\left( \mathrm{row}_i\left( \mathbf{G}\right)\right)\right)\right)$ (for any $1\le i\le \ell$ and $0\le \nu\le N-1$) is zero. Thus, $\mathrm{row}_j\left( \mathbf{H}\right)\in\mathcal{C}^\perp$ because of our discussion before the lemma.
\end{proof}

\begin{lemma}
\label{lemma2}
The matrix $\mathbf{H}$ is a GPM of a $\Delta$-MT code.
\end{lemma}
\begin{proof}
Our aim is to prove that $\mathbf{B}\mathbf{H}=\mathrm{diag}\left[x^{m_1}-\frac{1}{\lambda_1},\ldots, x^{m_\ell}-\frac{1}{\lambda_\ell}\right]$ for some polynomial matrix $\mathbf{B}$. Replacing $x$ with $\frac{1}{x}$ in \eqref{identical_eq} gives $\mathbf{A}\left(\frac{1}{x}\right)\mathbf{G}\left(\frac{1}{x}\right)=\mathbf{D}\left(\frac{1}{x}\right)$ as matrices over the ring $\mathbb{F}_q\left[x,\frac{1}{x}\right]$. From Definition \ref{def_parity}, 
\begin{equation*}
\mathrm{diag}\left[x^{m_1},\ldots,x^{m_\ell}\right]\mathbf{A}\left(\frac{1}{x}\right)=\mathbf{A}^*\mathrm{diag}\left[x^{d_1},\ldots,x^{d_\ell}\right].
\end{equation*}
Thus,
\begin{align*}
\mathbf{A}^* \mathrm{diag}\left[x^{d_1},\ldots,x^{d_\ell}\right]\mathbf{G}\left(\frac{1}{x}\right)&=\mathrm{diag}\left[x^{m_1},\ldots,x^{m_\ell}\right]\mathbf{D}\left(\frac{1}{x}\right)\\
&=-\mathrm{diag}\left[x^{m_1}-\frac{1}{\lambda_1},\ldots, x^{m_\ell}-\frac{1}{\lambda_\ell}\right] \mathrm{diag}\left[\lambda_1,\ldots,\lambda_\ell\right].
\end{align*}
For $1\le j\le \ell$, $\deg(a_{j,j})=m_j-d_j$. Thus, the diagonal elements of $\mathbf{A}^*$ have no negative powers of $x$. Again from Definition \ref{def_parity}, there is a strictly upper triangular matrix $\mathbf{S}$ such that 
\begin{equation*}
\mathbf{A}^{**}=\mathbf{A}^*+\mathrm{diag}\left[x^{m_1}-\frac{1}{\lambda_1},\ldots,x^{m_\ell}-\frac{1}{\lambda_\ell}\right] \mathbf{S}.
\end{equation*} 
Therefore,
\begin{align*}
&\mathbf{A}^{**}\mathrm{diag}\left[x^{d_1},\ldots,x^{d_\ell}\right]\mathbf{G}\left(\frac{1}{x}\right)\\
&\qquad\qquad=-\mathrm{diag}\left[x^{m_1}-\frac{1}{\lambda_1},\ldots, x^{m_\ell}-\frac{1}{\lambda_\ell}\right] \mathrm{diag}\left[\lambda_1,\ldots,\lambda_\ell\right]\\
&\qquad\qquad\quad+ \mathrm{diag}\left[x^{m_1}-\frac{1}{\lambda_1},\ldots,x^{m_\ell}-\frac{1}{\lambda_\ell}\right] \mathbf{S}\ \mathrm{diag}\left[x^{d_1},\ldots,x^{d_\ell}\right]\mathbf{G}\left(\frac{1}{x}\right)\\
&\qquad\qquad=\mathrm{diag}\left[x^{m_1}-\frac{1}{\lambda_1},\ldots, x^{m_\ell}-\frac{1}{\lambda_\ell}\right]\mathbf{U}
\end{align*}
where 
\begin{equation*}
\mathbf{U}=-\mathrm{diag}\left[\lambda_1,\ldots,\lambda_\ell\right]+ \mathbf{S}\ \mathrm{diag}\left[x^{d_1},\ldots,x^{d_\ell}\right]\mathbf{G}\left(\frac{1}{x}\right).
\end{equation*}
Note that $\mathbf{U}$ is an upper triangular invertible matrix because its determinant is a unit in $\mathbb{F}_q\left[x,\frac{1}{x}\right]$. Then,
\begin{equation*}
\mathbf{A}^{**}\mathrm{diag}\left[x^{d_1},\ldots,x^{d_\ell}\right]\mathbf{G}\left(\frac{1}{x}\right) \mathbf{U}^{-1}=\mathrm{diag}\left[x^{m_1}-\frac{1}{\lambda_1},\ldots, x^{m_\ell}-\frac{1}{\lambda_\ell}\right]
\end{equation*}
and
\begin{equation*}
\left( \mathrm{diag}\left[x^{d_1},\ldots,x^{d_\ell}\right]\mathbf{G}\left(\frac{1}{x}\right) \mathbf{U}^{-1}\right)^t \mathbf{H} =\mathrm{diag}\left[x^{m_1}-\frac{1}{\lambda_1},\ldots, x^{m_\ell}-\frac{1}{\lambda_\ell}\right].
\end{equation*}
Let 
\begin{equation*}
\mathbf{B}=\left( \mathrm{diag}\left[x^{d_1},\ldots,x^{d_\ell}\right]\mathbf{G}\left(\frac{1}{x}\right) \mathbf{U}^{-1}\right)^t.
\end{equation*}
Then $\mathbf{B}$ is lower triangular such that 
\begin{equation*}
\mathbf{B}\mathbf{H}=\mathrm{diag}\left[x^{m_1}-\frac{1}{\lambda_1},\ldots, x^{m_\ell}-\frac{1}{\lambda_\ell}\right].
\end{equation*}
The diagonal elements of $\mathbf{B}$ and $\mathbf{H}$ are polynomials with nonzero constant terms, thus the entries of $\mathbf{B}$ are elements of $\mathbb{F}_q[x]$. Therefore, $\mathbf{H}$ is a GPM for some $\Delta$-MT code.
\end{proof}

So far, Lemmas \ref{lemma1} and \ref{lemma2} show that $\mathbf{H}$ is a GPM of a $\Delta$-MT subcode of $\mathcal{C}^\perp$. Now we apply the standard dimension argument to show that this subcode is $\mathcal{C}^\perp$.

\begin{lemma}
\label{lemma3}
The matrix $\mathbf{H}$ is a GPM of a $\Delta$-MT code of dimension $n-k$ as an $\mathbb{F}_q$-subspace of $\mathbb{F}_q^n$, where $n=\sum_{j=1}^\ell m_j$ and $k$ is the dimension of $\mathcal{C}$.
\end{lemma}
\begin{proof}
Suppose $\mathbf{H}'$ is the Hermite normal form of $\mathbf{H}$ and let $\mathbf{B}'$ be the matrix that satisfies the identical equation of $\mathbf{H}'$. Then, there is an invertible polynomial matrix $\mathbf{U}$ such that $\mathbf{H}'=\mathbf{U}\mathbf{H}$. From Theorem \ref{dim_MT} and Definition \ref{def_parity}, the dimension of the subcode generated by $\mathbf{H}$ is 
\begin{equation*}\begin{split}
\deg\left(\mathrm{det}\left( \mathbf{B}'\right)\right)=n-\deg\left(\mathrm{det}\left( \mathbf{H}'\right)\right)&=n-\deg\left(\mathrm{det}\left( \mathbf{U}\right)\mathrm{det}\left( \mathbf{H}\right)\right)\\
&=n-\deg\left(\mathrm{det}\left( \mathbf{H}\right)\right)=n-\deg\left(\mathrm{det}\left( \mathbf{A}\right)\right)=n-k.
\end{split}\end{equation*}
\end{proof}

What we proved in Lemma \ref{lemma1}, Lemma \ref{lemma2}, and Lemma \ref{lemma3} can be summarized in the following theorem.
\begin{theorem}
\label{MT_Parity_H}
Let $\mathcal{C}$ be a $\Lambda$-MT code over $\mathbb{F}_q$ with reduced GPM $\mathbf{G}$ and let $\mathbf{A}$ be the matrix that satisfies the identical equation of $\mathbf{G}$. The polynomial matrix $\mathbf{H}$ given in Definition \ref{def_parity} is a GPM for $\mathcal{C}^\perp$.
\end{theorem}

\begin{example}
We continue with the $(2,1)$-MT code $\mathcal{C}$ discussed in Example \ref{ex_Sh2}. From Theorem \ref{MT_dual}, $\mathcal{C}^\perp$ is $(2,1)$-MT over $\mathbb{F}_3$ of length $60$ and dimension $54$. A GPM for $\mathcal{C}^\perp$ can be obtained from Definition \ref{def_parity} and Theorem \ref{MT_Parity_H} as follows:
\begin{equation*}
\mathbf{A}=\begin{pmatrix}
2+2 x + x^4 + x^5+x^6 & \quad 2x(1+x)^4\\
0 & 1
\end{pmatrix}.
\end{equation*}
\begin{equation*}
\begin{split}
\mathbf{A}\left( \frac{1}{x}\right)&=\begin{pmatrix}
x^{-6} + x^{-5} + x^{-4} +2 x^{-1} +2 & \quad 2x^{-1}(x^{-1}+1)^4\\
0 & 1
\end{pmatrix}.\\
\mathbf{A}^*&=\begin{pmatrix}
x^{20-14}\left(x^{-6} + x^{-5} + x^{-4} +2 x^{-1} +2\right) & \quad 2x^{20-40}x^{-1}(x^{-1}+1)^4\\
0 & x^{40-40}
\end{pmatrix}\\
&=\begin{pmatrix}
1 + x + x^{2} +2 x^{5} +2x^6 & \quad 2x^{-25}+ 2x^{-24}+ 2x^{-22}+2x^{-21}\\
0 & 1
\end{pmatrix}.\\
\mathbf{A}^{**}&=\begin{pmatrix}
1 + x + x^{2} +2 x^{5} +2x^6 & \quad x^{-5}+ x^{-4}+ x^{-2}+x^{-1}\\
0 & 1
\end{pmatrix}\\
&=\begin{pmatrix}
1 + x + x^{2} +2 x^{5} +2x^6 & \quad 2x^{15} + 2x^{16} + 2x^{18} + 2x^{19}\\
0 & 1
\end{pmatrix}.\\
\mathbf{H}&=\begin{pmatrix}
1 + x + x^{2} +2 x^{5} +2x^6 &  0 \\
2x^{15} + 2x^{16} + 2x^{18} + 2x^{19}& 1
\end{pmatrix}.
\end{split}
\end{equation*}
The reduced GPM of $\mathcal{C}^\perp$ is
\begin{equation*}
\mathbf{H}'=\begin{pmatrix}
1 & 2x+2x^2+x^3+x^4+x^5  \\
 0 &    2+2x+2x^2+x^5+x^6
\end{pmatrix}
\end{equation*}
which can be obtained by reducing $\mathbf{H}$ to its Hermite normal form.
\end{example}

Since the class of MT codes contains QC, QT, and GQC codes as subclasses, the following special cases are direct consequences of Theorem \ref{MT_Parity_H}.
\begin{corollary}
Let $\mathcal{C}$ be a QC code over $\mathbb{F}_q$ of index $\ell$, co-index $m$, and reduced GPM $\mathbf{G}=[g_{i,j}]$. Let $\mathbf{A}$ denote the matrix satisfying the identical equation of $\mathbf{G}$. Then $\mathcal{C}^\perp$ is QC of index $\ell$, co-index $m$, and a GPM 
\begin{equation*}
\mathbf{H}=\left( \mathbf{A}\left(\frac{1}{x}\right) \mathrm{diag}\left[x^{m-d_1},\ldots,x^{m-d_\ell}\right]\pmod{x^m-1}\right)^t
\end{equation*}
where $d_j=\deg\left(g_{j,j}\right)$ for $1\le j\le \ell$. 
\end{corollary}

\begin{corollary}
Let $\mathcal{C}$ be a QT code over $\mathbb{F}_q$ of index $\ell$, co-index $m$, shift constant $\lambda$, and reduced GPM $\mathbf{G}=[g_{i,j}]$. Let $\mathbf{A}$ denote the matrix satisfying the identical equation of $\mathbf{G}$. Then $\mathcal{C}^\perp$ is QT of index $\ell$, co-index $m$, shift constant $\frac{1}{\lambda}$, and a GPM 
\begin{equation*}
\mathbf{H}=\left( \mathbf{A}\left(\frac{1}{x}\right) \mathrm{diag}\left[x^{m-d_1},\ldots,x^{m-d_\ell}\right]\pmod{x^m-\frac{1}{\lambda}}\right)^t
\end{equation*}
where $d_j=\deg\left(g_{j,j}\right)$ for $1\le j\le \ell$. 
\end{corollary}

\begin{corollary}
Let $\mathcal{C}$ be a GQC code over $\mathbb{F}_q$ of index $\ell$, block lengths $(m_1,m_2,\ldots,m_\ell)$, and reduced GPM $\mathbf{G}=[g_{i,j}]$. Let $\mathbf{A}$ denote the matrix satisfying the identical equation of $\mathbf{G}$. Then $\mathcal{C}^\perp$ is GQC of index $\ell$, block lengths $(m_1,m_2,\ldots,m_\ell)$, and a GPM $\mathbf{H}$, where
\begin{equation*}
\mathrm{Column}_j\left(\mathbf{H}\right)=\mathrm{row}_j \left( \mathbf{A}\left(\frac{1}{x}\right) \mathrm{diag}\left[x^{m_j-d_1},\ldots,x^{m_j-d_\ell}\right]\pmod{x^{m_j}-1}\right)
\end{equation*}
and $d_j=\deg\left(g_{j,j}\right)$ for $1\le j\le \ell$. 
\end{corollary}

\section{Right, left, and two-sided Galois duals}
\label{Galois_dual_Sec}
In this section, we aim to generalize the result of Section \ref{Euclidean_dual} by replacing the Euclidean inner product with the Galois inner product. Furthermore, we present the two-sided Galois inner product of MT codes which has not been previously discussed in any study. Throughout this section, $q=p^e$ where $p$ is a prime and $e$ is a positive integer. Recall that the Frobenius automorphism of $\mathbb{F}_{q}$, denoted $\sigma$, is defined by $\sigma\left(\alpha \right)=\alpha^p$ for each $\alpha\in\mathbb{F}_{q}$. The Galois group of $\mathbb{F}_{q}$ is finite, cyclic, and generated by $\sigma$. The least positive integer $t$ such that $\sigma^t\left(\alpha\right)=\alpha$ for all $\alpha\in\mathbb{F}_{q}$ is $t=e$. Thus, the order of $\sigma$ in the Galois group is $e$. Then $\sigma$ extends to a ring automorphism of $\mathbb{F}_{q}[x]$ by defining it on polynomials over $\mathbb{F}_{q}$ as follows:
\begin{equation*}
\sigma\left(a_0+a_1 x +\cdots +a_m x^m\right) = \sigma\left(a_0\right)+\sigma\left(a_1\right) x +\cdots +\sigma\left(a_m\right) x^m.
\end{equation*}
In a natural way, $\sigma$ extends to an automorphism of the ring of matrices over $\mathbb{F}_{q}[x]$. For an $\ell\times\ell$ matrix $\mathbf{Y}=\left[y_{i,j}\right]$ over $\mathbb{F}_{q}[x]$, define 
\begin{equation*}
\sigma\left(\mathbf{Y}\right)=\left[\sigma\left(y_{i,j}\right)\right].
\end{equation*}
Let $0\le \mu< e$ and let $\upsilon=\mathrm{gcd}\left(e,\mu\right)$. Since $\upsilon$ divides $e$, $\mathbb{F}_{q}$ has $\mathbb{F}_{p^\upsilon}$ as a subfield. The automorphism $\sigma^\mu$ of $\mathbb{F}_{q}$ fixes an element $\alpha\in\mathbb{F}_{q}$, i.e., $\sigma^\mu\left(\alpha\right)=\alpha$, if and only if $\alpha\in\mathbb{F}_{p^\upsilon}$. In the ring of polynomials $\mathbb{F}_{q}[x]$, $\sigma^\mu$ only fixes all polynomials over $\mathbb{F}_{p^\upsilon}$. 
However, in the ring of matrices over $\mathbb{F}_{q}[x]$, $\sigma^\mu$ only fixes all matrices over $\mathbb{F}_{p^\upsilon}[x]$. Now, we examine the action of $\sigma^\mu$ on vectors of $\mathbb{F}_{q}^n$ as an introduction to study its effect on linear codes over $\mathbb{F}_{q}$ of length $n$.

\begin{definition}
Let $e$ be a positive integer and let $\sigma$ be the Frobenius automorphism of $\mathbb{F}_{q}$. For $0\le \mu< e$, the map $\sigma^\mu:\mathbb{F}_{q}^n\rightarrow \mathbb{F}_{q}^n$ is defined by
\begin{equation*}
\sigma^\mu: \left(a_1,a_2,\ldots,a_n\right) \mapsto \left(\sigma^\mu\left(a_1\right),\sigma^\mu\left(a_2\right),\ldots,\sigma^\mu\left(a_n\right)\right).
\end{equation*}
\end{definition}
Since $\sigma^\mu\left(\mathbf{a}+\mathbf{b}\right)=\sigma^\mu\left(\mathbf{a}\right)+\sigma^\mu\left(\mathbf{b}\right)$ for any $\mathbf{a},\mathbf{b}\in\mathbb{F}_{q}^n$, $\sigma^\mu$ defines an additive group automorphism on $\mathbb{F}_{q}^n$. For a linear code $\mathcal{C}$ over $\mathbb{F}_{q}$ of length $n$, let $\sigma^\mu\left(\mathcal{C}\right)=\left\{\sigma^\mu\left(\mathbf{c}\right) \ \mid \ \mathbf{c}\in\mathcal{C}\right\}$. The map $\sigma^\mu:\mathcal{C}\rightarrow \sigma^\mu\left(\mathcal{C}\right)$ is a group isomorphism; it is a group automorphism if and only if $\sigma^\mu\left(\mathcal{C}\right)=\mathcal{C}$, i.e., $\mathcal{C}$ is invariant under $\sigma^\mu$. We know that $\left(a_1,a_2,\ldots,a_n\right)\in\mathbb{F}_{q}^n$ is fixed under $\sigma^\mu$ if and only if $a_i\in\mathbb{F}_{p^\upsilon}$ for every $1\le i\le n$, where $\upsilon=\mathrm{gcd}\left(e,\mu\right)$. From the uniqueness of the reduced row echelon form of a generator matrix of a linear code, $\mathcal{C}$ is invariant under $\sigma^\mu$ if and only if $\mathcal{C}$ has a basis that is a subset of $\mathbb{F}_{p^\upsilon}^n$. Henceforth, for any $0\le \mu< e$ and some $\lambda_1, \lambda_2,\ldots, \lambda_\ell \in\mathbb{F}_{q}$, let 
\begin{equation*}
\sigma^\mu\left(\Lambda\right)=\left(\sigma^\mu\left(\lambda_1\right),\sigma^\mu\left(\lambda_2\right),\ldots,\sigma^\mu\left(\lambda_\ell\right) \right)
\end{equation*}
and
\begin{equation*}
\sigma^\mu\left(\Delta\right)=\left(\sigma^\mu\left(\frac{1}{\lambda_1}\right),\sigma^\mu\left(\frac{1}{\lambda_2}\right),\ldots,\sigma^\mu\left(\frac{1}{\lambda_\ell}\right) \right).
\end{equation*}

\begin{theorem}
\label{Code_map}
Let $\mathcal{C}$ be a linear code over $\mathbb{F}_{q}$ of length $n$ and dimension $k$. For any $0\le \mu <e$, $\sigma^\mu\left(\mathcal{C}\right)$ is a linear code over $\mathbb{F}_{q}$ of length $n$ and dimension $k$. Moreover, $\mathcal{C}$ and $\sigma^\mu\left(\mathcal{C}\right)$ are isomorphic as additive groups. Suppose $\mathcal{C}$ is $\Lambda$-MT of block lengths $\left(m_1,m_2,\ldots,m_\ell\right)$. Then $\sigma^\mu\left(\mathcal{C}\right)$ is $\sigma^\mu\left(\Lambda \right)$-MT of block lengths $\left(m_1,m_2,\ldots,m_\ell\right)$. Let $\mathbf{G}$ be the reduced GPM of $\mathcal{C}$ and $\mathbf{A}$ the matrix that satisfies the identical equation of $\mathbf{G}$. Then the reduced GPM of $\sigma^\mu\left(\mathcal{C}\right)$ is $\sigma^\mu\left(\mathbf{G}\right)$ and $\sigma^\mu\left(\mathbf{A}\right)$ is the matrix that satisfies the identical equation of $\sigma^\mu\left(\mathbf{G}\right)$.
\end{theorem}
\begin{proof}
Since $\sigma^\mu\left(\mathcal{C}\right)$ is the image of $\mathcal{C}$ under a group isomorphism, $\sigma^\mu\left(\mathcal{C}\right)$ is an abelian group. Also, $\sigma^\mu\left(\mathcal{C}\right)$ is linear over $\mathbb{F}_{q}$ because $\alpha\sigma^\mu\left( \mathbf{c}\right)=\sigma^\mu\left(\sigma^{e-\mu}\left(\alpha\right) \mathbf{c}\right)\in\sigma^\mu\left(\mathcal{C}\right)$ for every $\alpha\in\mathbb{F}_{q}$ and $\mathbf{c}\in\mathcal{C}$. In addition, $\sigma^\mu\left(\mathcal{C}\right)$ has dimension $k$ because $\mid\mathcal{C}\mid=\mid\sigma^\mu\left(\mathcal{C}\right)\mid$. Let $\mathcal{C}$ be $\Lambda$-MT, then for any $\mathbf{c}\in\mathcal{C}$, we have
\begin{equation*}
T_{\sigma^\mu\left(\Lambda\right)}\left(\sigma^\mu\left( \mathbf{c}\right) \right)=\sigma^\mu\left( T_{\Lambda}\left(\mathbf{c}\right) \right)\in \sigma^\mu\left(\mathcal{C}\right).
\end{equation*}
Thus $\sigma^\mu\left(\mathcal{C}\right)$ is $T_{\sigma^\mu\left(\Lambda\right)}$-invariant. Applying $\sigma^\mu$ to the matrix equation \eqref{identical_eq} yields 
\begin{equation}
\label{Identical_eq_in_proof}
\sigma^\mu\left(\mathbf{A}\right)\sigma^\mu\left(\mathbf{G}\right)=\sigma^\mu\left(\mathbf{D}\right)=\mathrm{diag}\left[x^{m_1}-\sigma^\mu\left(\lambda_1\right),\ldots,x^{m_\ell}-\sigma^\mu\left(\lambda_\ell\right)\right]. 
\end{equation}
Therefore, $\sigma^\mu\left(\mathbf{G}\right)$ is a GPM of a $\sigma^\mu\left(\Lambda\right)$-MT code $\mathcal{C}'$ over $\mathbb{F}_{q}$ of block lengths $\left(m_1,m_2,\ldots,m_\ell\right)$. Since $\sigma^\mu$ preserves the degree of the polynomials, we see that $\sigma^\mu\left(\mathbf{G}\right)$ is in Hermite normal form and $\mathcal{C}'$ has dimension $k$. 
Furthermore, $\mathcal{C}'\subseteq \sigma^\mu\left(\mathcal{C}\right)$ since the rows of $\sigma^\mu\left(\mathbf{G}\right)$ are codewords in $\sigma^\mu\left(\mathcal{C}\right)$. By the standard dimension argument, $\mathcal{C}' = \sigma^\mu\left(\mathcal{C}\right)$. Thus $\sigma^\mu\left(\mathbf{G}\right)$ is the reduced GPM of $\sigma^\mu\left(\mathcal{C}\right)$. Specifically, \eqref{Identical_eq_in_proof} shows that $\sigma^\mu\left(\mathbf{A}\right)$ is the matrix that satisfies the identical equation of $\sigma^\mu\left(\mathbf{G}\right)$.
\end{proof}

In \cite{Fan2016}, the Galois inner product is defined as an intrinsic generalization of the Euclidean inner product. Let $\mathbf{a}=\left(a_1,a_2,\ldots,a_n\right)$ and $\mathbf{b}=\left(b_1,b_2,\ldots,b_n\right)$ be two vectors in $\mathbb{F}_{q}^n$. The Euclidean inner product of $\mathbf{a}$ and $\mathbf{b}$, denoted $\langle\mathbf{a},\mathbf{b}\rangle$, is a symmetric bilinear form. For a fixed non-negative integer $\kappa <e$, define the $\kappa$-Galois inner product of $\mathbf{a}$ and $\mathbf{b}$ by the formula
\begin{equation*}
\langle \mathbf{a},\mathbf{b}\rangle_\kappa=\sum_{i=1}^n a_i \sigma^\kappa\left(b_i\right)=\sum_{i=1}^n a_i b_i^{p^\kappa}.
\end{equation*}
Clearly, $\langle \mathbf{a},\mathbf{b}\rangle_\kappa$ and $\langle \mathbf{b},\mathbf{a}\rangle_\kappa$ are not necessarily equal. For any $\alpha\in\mathbb{F}_{q}$, $\langle \alpha\mathbf{a},\mathbf{b}\rangle_\kappa=\alpha \langle \mathbf{a},\mathbf{b}\rangle_\kappa$ and $\langle \mathbf{a},\alpha\mathbf{b}\rangle_\kappa=\alpha^{p^\kappa}\langle \mathbf{a},\mathbf{b}\rangle_\kappa$. Thus, the Galois inner product is neither symmetric nor bilinear. In fact, $\langle \mathbf{a},\mathbf{b}\rangle=\langle \mathbf{a},\mathbf{b}\rangle_0$ for every $\mathbf{a},\mathbf{b}\in\mathbb{F}_{q}^n$. For a linear code $\mathcal{C}$ over $\mathbb{F}_{q}$ of length $n$, the Euclidean inner product is used to define the Euclidean dual $\mathcal{C}^\perp$ of $\mathcal{C}$. On using the same imitation for the Galois inner product, we have to define two distinct duals: the right $\kappa$-Galois dual $\mathcal{C}^{\perp_\kappa}$ and the $\kappa$-left Galois dual $\mathcal{C}_{\perp_\kappa}$.

\begin{definition}
\label{right_dual_def}
Let $\kappa<e$ be a non-negative integer and let $\mathcal{C}$ be a linear code over $\mathbb{F}_{q}$ of length $n$. Define the right $\kappa$-Galois dual of $\mathcal{C}$ as follows:
\begin{equation*}
\mathcal{C}^{\perp_\kappa}=\left\{ \mathbf{a}\in\mathbb{F}_{q}^n \ \mid\  \langle \mathbf{c},\mathbf{a} \rangle_\kappa=0 \ \forall\  \mathbf{c}\in\mathcal{C}\right\}.
\end{equation*}
\end{definition}

\begin{theorem}
\label{right_dual}
Let $\kappa<e$ be a non-negative integer and let $\mathcal{C}$ be a linear code over $\mathbb{F}_{q}$ of length $n$ and dimension $k$. The right $\kappa$-Galois dual $\mathcal{C}^{\perp_\kappa}$ of $\mathcal{C}$ is a linear code over $\mathbb{F}_{q}$ of length $n$ and dimension $n-k$. In fact, $\mathcal{C}^{\perp_\kappa}=\sigma^{e-\kappa}\left(\mathcal{C}^\perp\right)=\left(\sigma^{e-\kappa}\mathcal{C}\right)^\perp$. Suppose that $\mathcal{C}$ is $\Lambda$-MT of block lengths $\left(m_1,m_2,\ldots,m_\ell\right)$. Then $\mathcal{C}^{\perp_\kappa}$ is $\sigma^{e-\kappa}\left(\Delta\right)$-MT of block lengths $\left(m_1,m_2,\ldots,m_\ell\right)$ and has the reduced GPM $\sigma^{e-\kappa}\left(\mathbf{H}\right)$, where $\mathbf{H}$ is the reduced GPM of $\mathcal{C}^\perp$.
\end{theorem}
\begin{proof}
Observe that $\langle \mathbf{a},\mathbf{b}\rangle_\kappa=\langle \mathbf{a},\sigma^{\kappa}\mathbf{b}\rangle$ for any $\mathbf{a},\mathbf{b}\in\mathbb{F}_{q}^n$. Thus $\mathcal{C}^{\perp_\kappa}=\sigma^{e-\kappa}\left(\mathcal{C}^\perp\right)$. Also observe that $\langle \sigma^{e-\kappa}\mathbf{a},\mathbf{b}\rangle=0$ if and only if $\langle \mathbf{a},\sigma^{\kappa}\mathbf{b}\rangle=0$. Thus $\sigma^{e-\kappa}\left(\mathcal{C}^\perp\right)=\left(\sigma^{e-\kappa}\mathcal{C}\right)^\perp$. Since $\sigma^{e-\kappa}: \mathcal{C}^\perp \rightarrow \mathcal{C}^{\perp_\kappa}$ is a group isomorphism, $\mathcal{C}^{\perp_\kappa}$ is linear of dimension $n-k$ by Theorem \ref{Code_map}. Now assume that $\mathcal{C}$ is $\Lambda$-MT and recall from Theorem \ref{MT_dual} that $\mathcal{C}^\perp$ is $\Delta$-MT. Then by Theorem \ref{Code_map}, $\mathcal{C}^{\perp_\kappa}$ is $\sigma^{e-\kappa}\left(\Delta\right)$-MT with reduced GPM $\sigma^{e-\kappa}\left(\mathbf{H}\right)$.
\end{proof}

Analogously to Definition \ref{right_dual_def} and Theorem \ref{right_dual}, we define and investigate the left $\kappa$-Galois dual.
\begin{definition}
Let $\kappa<e$ be a non-negative integer and let $\mathcal{C}$ be a linear code over $\mathbb{F}_{q}$ of length $n$. The left $\kappa$-Galois dual $\mathcal{C}_{\perp_\kappa}$ of $\mathcal{C}$ is defined as the subset of $\mathbb{F}_{q}^n$ for which $\left(\mathcal{C}_{\perp_\kappa}\right)^{\perp_\kappa}=\mathcal{C}$, or in other words,
\begin{equation*}
\mathcal{C}_{\perp_\kappa}=\left\{ \mathbf{a}\in\mathbb{F}_{q}^n \ \mid\  \langle \mathbf{a},\mathbf{c} \rangle_\kappa=0 \ \forall\  \mathbf{c}\in\mathcal{C}\right\}.
\end{equation*}
\end{definition}

\begin{theorem}
\label{left_dual}
Let $\kappa<e$ be a non-negative integer and let $\mathcal{C}$ be a linear code over $\mathbb{F}_{q}$ of length $n$ and dimension $k$.
The left $\kappa$-Galois dual $\mathcal{C}_{\perp_\kappa}$ of $\mathcal{C}$ is a linear code over $\mathbb{F}_{q}$ of length $n$ and dimension $n-k$. In fact, $\mathcal{C}_{\perp_\kappa}=\sigma^{\kappa}\left(\mathcal{C}^\perp\right)=\left(\sigma^{\kappa}\mathcal{C}\right)^\perp$. Suppose that $\mathcal{C}$ is $\Lambda$-MT of block lengths $\left(m_1,m_2,\ldots,m_\ell\right)$. Then $\mathcal{C}_{\perp_\kappa}$ is $\sigma^\kappa\left(\Delta\right)$-MT of block lengths $\left(m_1,m_2,\ldots,m_\ell\right)$ and has the reduced GPM $\sigma^{\kappa}\left(\mathbf{H}\right)$, where $\mathbf{H}$ is the reduced GPM of $\mathcal{C}^\perp$.
\end{theorem}
\begin{proof}
Observe that $\langle \mathbf{a},\mathbf{b}\rangle_\kappa=\langle \mathbf{a},\sigma^\kappa\mathbf{b}\rangle=\sigma^\kappa \left(\langle \sigma^{e-\kappa}\mathbf{a},\mathbf{b}\rangle\right)$ for any $\mathbf{a},\mathbf{b}\in\mathbb{F}_{q}^n$. Thus $\mathcal{C}_{\perp_\kappa}=\sigma^{\kappa}\left(\mathcal{C}^\perp\right)=\left(\sigma^{\kappa}\mathcal{C}\right)^\perp$. Since $\sigma^{\kappa}: \mathcal{C}^\perp \rightarrow \mathcal{C}_{\perp_\kappa}$ is a group isomorphism, $\mathcal{C}_{\perp_\kappa}$ is linear of dimension $n-k$ by Theorem \ref{Code_map}. The last part can be proven in a way similar to that of Theorem \ref{right_dual}.
\end{proof}

In the literature and to the extent of our knowledge, left and right Galois duals were not jointly discussed in the same study. For instance, the Galois dual of a linear code is defined in \cite{Hongwei} similar to our definition of the right Galois dual. Whereas in \cite{Sharma}, the Galois dual is defined similar to our definition of the left Galois dual. We found it useful to include both Galois duals so that we can examine their interrelationships. One of these benefits is the following result.

\begin{theorem}
\label{Properties}
Let $\kappa<e$ be a non-negative integer and let $\mathcal{C}$ be a linear code over $\mathbb{F}_{q}$. Then 
\begin{enumerate}
\item $\left(\mathcal{C}^{\perp_\kappa}\right)_{\perp_\kappa}=\mathcal{C}$.
\item $\left( \sigma^\kappa \mathcal{C}\right)^{\perp_\kappa}=\mathcal{C}^{\perp}=\sigma^\kappa\left(  \mathcal{C}^{\perp_\kappa}\right)$.\label{item2}
\item $\left( \sigma^{e-\kappa} \mathcal{C}\right)_{\perp_\kappa}=\mathcal{C}^{\perp}=\sigma^{e-\kappa}\left(  \mathcal{C}_{\perp_\kappa}\right)$.
\item $\mathcal{C}^{\perp_\kappa}=\left( \sigma^{2(e-\kappa)} \mathcal{C}\right)_{\perp_\kappa}=\sigma^{2(e-\kappa)}\left(  \mathcal{C}_{\perp_\kappa}\right)$.
\item $\mathcal{C}_{\perp_\kappa}=\left( \sigma^{2\kappa} \mathcal{C}\right)^{\perp_\kappa}=\sigma^{2\kappa}\left(  \mathcal{C}^{\perp_\kappa}\right)$.
\item $\mathcal{C}^{\perp_\kappa}=\mathcal{C}_{\perp_\kappa}$ if and only if $\sigma^{2\kappa}\left(\mathcal{C}\right)=\mathcal{C}$.
\end{enumerate}
\end{theorem}
\begin{proof}
By Theorems \ref{right_dual} and \ref{left_dual}, we have
\begin{enumerate}
\item $\left(\mathcal{C}^{\perp_\kappa}\right)_{\perp_\kappa}=\left(\left(\sigma^{e-\kappa}\mathcal{C}\right)^\perp\right)_{\perp_\kappa}=\sigma^{\kappa}\left(\left(\sigma^{e-\kappa}\mathcal{C}\right)^{\perp\perp}\right)=\sigma^{\kappa}\left(\sigma^{e-\kappa}\mathcal{C}\right)=\mathcal{C}$.
\item $\left( \sigma^\kappa \mathcal{C}\right)^{\perp_\kappa}=\left( \sigma^{e-\kappa}\left(\sigma^\kappa \mathcal{C}\right)\right)^{\perp}=\mathcal{C}^{\perp}$. But, by Theorem \ref{right_dual}, $\mathcal{C}^\perp=\sigma^{\kappa}\left(\mathcal{C}^{\perp_\kappa}\right)$.
\item $\left( \sigma^{e-\kappa} \mathcal{C}\right)_{\perp_\kappa}=\left( \sigma^\kappa\left(\sigma^{e-\kappa} \mathcal{C}\right)\right)^{\perp}=\mathcal{C}^{\perp}$. But, by Theorem \ref{left_dual}, $\mathcal{C}^\perp=\sigma^{e-\kappa}\left(\mathcal{C}_{\perp_\kappa}\right)$.
\item The result follows by replacing $\mathcal{C}$ with $\sigma^{e-\kappa} \mathcal{C}$ in $\left( \sigma^\kappa \mathcal{C}\right)^{\perp_\kappa}=\left( \sigma^{e-\kappa} \mathcal{C}\right)_{\perp_\kappa}=\sigma^{e-\kappa}\left(  \mathcal{C}_{\perp_\kappa}\right)$ and using $\left( \sigma^{e-\kappa} \mathcal{C}\right)_{\perp_\kappa}=\sigma^{e-\kappa}\left(  \mathcal{C}_{\perp_\kappa}\right)$.
\item The result follows by replacing $\mathcal{C}$ with $\sigma^\kappa \mathcal{C}$ in $\left( \sigma^{e-\kappa} \mathcal{C}\right)_{\perp_\kappa}=\left( \sigma^\kappa \mathcal{C}\right)^{\perp_\kappa}=\sigma^\kappa\left(  \mathcal{C}^{\perp_\kappa}\right)$ and using $\left( \sigma^\kappa \mathcal{C}\right)^{\perp_\kappa}=\sigma^\kappa\left(  \mathcal{C}^{\perp_\kappa}\right)$.
\item Assume that $\sigma^{2\kappa}\left(\mathcal{C}\right)=\mathcal{C}$. Then $\mathcal{C}_{\perp_\kappa}=\left( \sigma^{2\kappa} \mathcal{C}\right)^{\perp_\kappa}=\mathcal{C}^{\perp_\kappa}$. Conversely, assume that $\mathcal{C}_{\perp_\kappa}=\mathcal{C}^{\perp_\kappa}$. It follows that $\left( \sigma^{2\kappa} \mathcal{C}\right)^{\perp_\kappa}=\mathcal{C}^{\perp_\kappa}$, hence
\[\mathcal{C}=\left(\mathcal{C}^{\perp_\kappa}\right)_{\perp_\kappa}=\left(\left( \sigma^{2\kappa} \mathcal{C}\right)^{\perp_\kappa}\right)_{\perp_\kappa}=\sigma^{2\kappa} \left(\mathcal{C}\right).\]
\end{enumerate}
\end{proof}

Throughout Theorems \ref{Maiin1}--\ref{Maiin6}, we will adopt the following notations without introducing them. Let $\kappa<e$ be a non-negative integer and choose a positive integer $\tau$ such that $e\mid4\kappa\tau$. Let $\lambda_j\in\mathbb{F}_{p^\upsilon}$ for $1\le j\le \ell$, where $\upsilon=\mathrm{gcd}\left(e,2\kappa\tau\right)$. Let $\mathcal{C}$ be a $\Lambda$-MT code over $\mathbb{F}_{q}$ of block lengths $\left(m_1,m_2,\ldots,m_\ell\right)$. The reduced GPM of $\mathcal{C}$ is $\mathbf{G}$, while $\mathbf{A}$ is the matrix that satisfies the identical equation of $\mathbf{G}$. The reduced GPM of $\mathcal{C}^\perp$ is $\mathbf{H}$, while $\mathbf{B}$ is the matrix that satisfies the identical equation of $\mathbf{H}$. Our first goal is to provide some results for the codes $\mathcal{C}^{\perp_\kappa}$ and $\sigma^{2\kappa\tau}\left(\mathcal{C}^{\perp_\kappa}\right)$.

\begin{theorem}
\label{Maiin1}
Both $\mathcal{C}^{\perp_\kappa}$ and $\sigma^{2\kappa\tau}\left(\mathcal{C}^{\perp_\kappa}\right)$ are $\sigma^{e-\kappa}\left(\Delta\right)$-MT codes over $\mathbb{F}_{q}$ of block lengths $\left(m_1,m_2,\ldots,m_\ell\right)$.
\end{theorem}
\begin{proof}
Since $\sigma^{2\kappa\tau}:\mathcal{C}^{\perp_\kappa}\rightarrow \sigma^{2\kappa\tau}\left(\mathcal{C}^{\perp_\kappa}\right)$ is a group isomorphism, Theorems \ref{Code_map} and \ref{right_dual} assert that $\sigma^{2\kappa\tau}\left(\mathcal{C}^{\perp_\kappa}\right)$ is a $\sigma^{2\kappa\tau}\left(\sigma^{e-\kappa}\left(\Delta\right)\right)$-MT code. In fact, $\sigma^{2\kappa\tau}\left(\sigma^{e-\kappa}\left(\Delta\right)\right)=\sigma^{e-\kappa}\left(\Delta\right)$ because $\sigma^{e-\kappa}\left(\Delta\right)\in\mathbb{F}_{p^\upsilon}^\ell$ and $\sigma^{2\kappa\tau}$ fixes all elements of $\mathbb{F}_{p^\upsilon}$.
\end{proof}

\begin{theorem}
\label{Maiin2}
The reduced GPM of $\mathcal{C}^{\perp_\kappa}$ is $\sigma^{e-\kappa}\left(\mathbf{H}\right)$, while the reduced GPM of $\sigma^{2\kappa\tau}\left(\mathcal{C}^{\perp_\kappa}\right)$ is $\sigma^{\kappa(2\tau-1)}\left(\mathbf{H}\right)$.
\end{theorem}
\begin{proof}
This is evident from Theorem \ref{Code_map} after noticing that $\mathcal{C}^{\perp_\kappa}=\sigma^{e-\kappa}\left(\mathcal{C}^\perp\right)$ and $\sigma^{2\kappa\tau}\left(\mathcal{C}^{\perp_\kappa}\right)=\sigma^{2\kappa\tau}\left(\sigma^{e-\kappa}\left(\mathcal{C}^\perp\right)\right)=\sigma^{\kappa(2\tau-1)}\left(\mathcal{C}^\perp\right)$.
\end{proof}
There is another result to be obtained from Theorem \ref{Maiin2} by applying $\sigma^{e-\kappa}$ and $\sigma^{\kappa(2\tau-1)}$ to the identical equation of $\mathbf{H}$. Namely, $\sigma^{e-\kappa}\left(\mathbf{B}\right)$ and $\sigma^{\kappa(2\tau-1)}\left(\mathbf{B}\right)$  are the two matrices that satisfy the identical equations of $\sigma^{e-\kappa}\left(\mathbf{H}\right)$ and $\sigma^{\kappa(2\tau-1)}\left(\mathbf{H}\right)$ respectively.

\begin{theorem}
\label{Maiin3}
The following conditions are equivalent.
\begin{enumerate}
\item $\mathcal{C}^{\perp_\kappa}=\sigma^{2\kappa\tau}\left(\mathcal{C}^{\perp_\kappa}\right)$.
\item $\sigma^{2\kappa\tau}\left(\mathcal{C}\right)=\mathcal{C}$.
\item $\mathbf{G}$ is a matrix over $\mathbb{F}_{p^\upsilon}[x]$. 
\end{enumerate}
\end{theorem}
\begin{proof}
Applying Theorem \ref{Properties}(\ref{item2}) recursively yields $\sigma^{2\kappa\tau}\left(\mathcal{C}^{\perp_\kappa}\right)=\left(\sigma^{2\kappa\tau}\mathcal{C}\right)^{\perp_\kappa}$. Assume that $\mathcal{C}^{\perp_\kappa}=\sigma^{2\kappa\tau}\left(\mathcal{C}^{\perp_\kappa}\right)$. Then 
\[\mathcal{C}=\left(\mathcal{C}^{\perp_\kappa}\right)_{\perp_\kappa}=\left(\sigma^{2\kappa\tau}\left(\mathcal{C}^{\perp_\kappa}\right)\right)_{\perp_\kappa}=\left(\left(\sigma^{2\kappa\tau}\mathcal{C}\right)^{\perp_\kappa}\right)_{\perp_\kappa}=\sigma^{2\kappa\tau}\left(\mathcal{C}\right).\]
Conversely, assume that $\sigma^{2\kappa\tau}\left(\mathcal{C}\right)=\mathcal{C}$. Then $\mathcal{C}^{\perp_\kappa}=\left(\sigma^{2\kappa\tau}\mathcal{C}\right)^{\perp_\kappa}=\sigma^{2\kappa\tau}\left(\mathcal{C}^{\perp_\kappa}\right)$.

By Theorem \ref{Code_map}, $\sigma^{2\kappa\tau}\left(\mathbf{G}\right)$ is the reduced GPM of $\sigma^{2\kappa\tau}\left(\mathcal{C}\right)$. By the uniqueness of the reduced GPM, $\mathcal{C}=\sigma^{2\kappa\tau}\left(\mathcal{C}\right)$ if and only if $\sigma^{2\kappa\tau}\left(\mathbf{G}\right)=\mathbf{G}$. Writing $\mathbf{G}=\left[g_{i,j}\right]$ where $g_{i,j}\in\mathbb{F}_{q}[x]$ for $1\le i,j\le \ell$, then $\sigma^{2\kappa\tau}\left(\mathbf{G}\right)=\mathbf{G}$ if and only if $\sigma^{2\kappa\tau}$ fixes $g_{i,j}$ for all $i,j$. That is, $\mathcal{C}=\sigma^{2\kappa\tau}\left(\mathcal{C}\right)$ if and only if $g_{i,j}\in\mathbb{F}_{p^\upsilon}[x]$ for all $i,j$. 
\end{proof}

\begin{theorem}
\label{Maiin4}
The code $\mathcal{C}^{\perp_\kappa} \cap \sigma^{2\kappa\tau}\left(\mathcal{C}^{\perp_\kappa}\right)$ is $\sigma^{e-\kappa}\left(\Delta\right)$-MT of block lengths $\left(m_1,m_2,\ldots,m_\ell\right)$ and is invariant under $\sigma^{2\kappa\tau}$.
\end{theorem}
\begin{proof}
The first result is immediate from Theorem \ref{Maiin1}. Note that $\sigma^{4\kappa\tau}$ acts as the identity on $\mathbb{F}_{q}^n$ because $e\mid4\kappa\tau$. Hence, 
\begin{equation*}
\sigma^{2\kappa\tau}\left(\mathcal{C}^{\perp_\kappa} \cap \sigma^{2\kappa\tau}\left(\mathcal{C}^{\perp_\kappa}\right)\right)\subseteq \mathcal{C}^{\perp_\kappa} \cap \sigma^{2\kappa\tau}\left(\mathcal{C}^{\perp_\kappa}\right).
\end{equation*}
This inequality turns into equality because $\sigma^{2\kappa\tau}$ is a group isomorphism. Hence, $\mathcal{C}^{\perp_\kappa} \cap \sigma^{2\kappa\tau}\left(\mathcal{C}^{\perp_\kappa}\right)$ is $\sigma^{2\kappa\tau}$-invariant.
\end{proof}

\begin{theorem}
\label{Maiin5}
Let $\mathcal{S}$ be a $\sigma^{e-\kappa}\left(\Delta\right)$-MT subcode of $\mathcal{C}^{\perp_\kappa} \cap \sigma^{2\kappa\tau}\left(\mathcal{C}^{\perp_\kappa}\right)$. Then $\mathcal{S}$ is $\sigma^{2\kappa\tau}$-invariant if and only if there exist upper triangular matrices $\mathbf{X}$ and $\mathbf{Y}$ over $\mathbb{F}_{p^\upsilon}[x]$ and $\mathbb{F}_{q}[x]$, respectively, such that
\begin{enumerate}
\item $\mathbf{Y}\sigma^{e-\kappa}\left( \mathbf{H}\right)$ is a GPM of $\mathcal{S}$,
\item $\mathbf{Y}\sigma^{e-\kappa}\left( \mathbf{H}\right)$ is a matrix over $\mathbb{F}_{p^\upsilon}[x]$, and
\item $\mathbf{X}\mathbf{Y}=\sigma^{e-\kappa}\left( \mathbf{B}\right)$.
\end{enumerate}
In this case, $\mathcal{S}$ has dimension $\deg\left(\mathrm{det}\left(\mathbf{X} \right)\right)$.
\end{theorem}
\begin{proof}
Assume that $\mathcal{S}$ is invariant under $\sigma^{2\kappa\tau}$. Let $\mathbf{P}$ be the reduced GPM of $\mathcal{S}$ and let $\mathbf{X}$ be the matrix that satisfies the identical equation of $\mathbf{P}$. By Theorems \ref{Containment} and \ref{Maiin2}, there exists an upper triangular matrix $\mathbf{Y}$ such that $\mathbf{P}=\mathbf{Y}\sigma^{e-\kappa}\left( \mathbf{H}\right)$ because $\mathcal{S}\subseteq \mathcal{C}^{\perp_\kappa}$. Since $\mathcal{S}$ is $\sigma^{2\kappa\tau}$-invariant, $\sigma^{2\kappa\tau}$ fixes $\mathbf{P}$. Thus $\mathbf{P}$ is over $\mathbb{F}_{p^\upsilon}[x]$, then so is $\mathbf{X}$ because $\lambda_j\in\mathbb{F}_{p^\upsilon}$ for $1\le j\le \ell$. But $\mathbf{X}\mathbf{Y}\sigma^{e-\kappa}\left( \mathbf{H}\right)=\mathbf{X}\mathbf{P}=\sigma^{e-\kappa}\left( \mathbf{B}\right)\sigma^{e-\kappa}\left( \mathbf{H}\right)$, then $\mathbf{X}\mathbf{Y}=\sigma^{e-\kappa}\left( \mathbf{B}\right)$. The dimension of $\mathcal{S}$ is immediate from Theorem \ref{dim_MT}.

Conversely, suppose that $\mathcal{S}$ has a GPM $\mathbf{Y}\sigma^{e-\kappa}\left( \mathbf{H}\right)$ over $\mathbb{F}_{p^\upsilon}[x]$. We know that $\sigma^{2\kappa\tau}$ fixes all matrices over $\mathbb{F}_{p^\upsilon}[x]$. Then $\mathbf{Y}\sigma^{e-\kappa}\left( \mathbf{H}\right)$ is a GPM for $\sigma^{2\kappa\tau}\left(\mathcal{S}\right)$. That is, $\sigma^{2\kappa\tau}\left(\mathcal{S}\right)=\mathcal{S}$.
\end{proof}

In Theorem \ref{Maiin6}, we continue to use the general setting used above; however, the main result of this section will be an immediate consequence of it. Specifically, setting $\tau=1$ yields a result that fits the two-sided Galois dual of a MT code. This is described in Corollary \ref {main_coroll}.

\begin{theorem}
\label{Maiin6}
Let $\mathbf{X}$ and $\mathbf{Y}$ be upper triangular matrices over $\mathbb{F}_{p^\upsilon}[x]$ and $\mathbb{F}_{q}[x]$, respectively, such that
\begin{enumerate}
\item $\mathbf{Y}\sigma^{e-\kappa}\left( \mathbf{H}\right)$ is a matrix over $\mathbb{F}_{p^\upsilon}[x]$, 
\item $\mathbf{X}\mathbf{Y}=\sigma^{e-\kappa}\left( \mathbf{B}\right)$, and
\item $\deg\left(\mathrm{det}\left(\mathbf{X}\right)\right)$ is maximum among all matrices that satisfy these conditions, or equivalently, $\deg\left(\mathrm{det}\left(\mathbf{Y}\right)\right)$ is minimum among all matrices that satisfy these conditions.
\end{enumerate}
Then $\mathcal{C}^{\perp_\kappa} \cap \sigma^{2\kappa\tau}\left(\mathcal{C}^{\perp_\kappa}\right)$ has a GPM $\mathbf{Y}\sigma^{e-\kappa}\left( \mathbf{H}\right)$ and dimension $\deg\left(\mathrm{det}\left(\mathbf{X}\right)\right)$. Furthermore, $\mathbf{X}$ is the matrix that satisfies the identical equation of $\mathbf{Y}\sigma^{e-\kappa}\left( \mathbf{H}\right)$.
\end{theorem}
\begin{proof}
Theorem \ref{Maiin5} and the maximality of $\deg\left(\mathrm{det}\left(\mathbf{X}\right)\right)$ allow us to conclude that $\mathbf{Y}\sigma^{e-\kappa}\left( \mathbf{H}\right)$ is a GPM for the largest $\sigma^{2\kappa\tau}$-invariant $\sigma^{e-\kappa}\left(\Delta\right)$-MT subcode of $\mathcal{C}^{\perp_\kappa} \cap \sigma^{2\kappa\tau}\left(\mathcal{C}^{\perp_\kappa}\right)$. Certainly, this is $\mathcal{C}^{\perp_\kappa} \cap \sigma^{2\kappa\tau}\left(\mathcal{C}^{\perp_\kappa}\right)$ by Theorem \ref{Maiin4}.
\end{proof}

Let $\mathcal{C}$ be a linear code. We define the two-sided Galois dual of $\mathcal{C}$ to be the intersection of its right and left Galois duals. Since $\mathcal{C}$ is linear, its two-sided Galois dual is linear because both $\mathcal{C}^{\perp_\kappa}$ and $\mathcal{C}_{\perp_\kappa}$ are linear. But if $\mathcal{C}$ is MT, then the two-sided Galois dual is not necessarily MT. Specifically, if $\mathcal{C}$ is $\Lambda$-MT, then $\mathcal{C}^{\perp_\kappa}$ is $\sigma^{e-\kappa}\left(\Delta\right)$-MT and $\mathcal{C}_{\perp_\kappa}$ is $\sigma^{\kappa}\left(\Delta\right)$-MT. A sufficient condition to ensure that the two-sided Galois dual is MT is $\sigma^{e-\kappa}\left(\Delta\right)=\sigma^{\kappa}\left(\Delta\right)$, or equivalently, $\lambda_j\in\mathbb{F}_{p^\upsilon}$ for $1\le j\le \ell$, where $\upsilon=\mathrm{gcd}\left(e,2\kappa\right)$. Motivated by this condition, we start with the following definition for the two-sided Galois dual of a MT code.

\begin{definition}
\label{2sided}
Let $e$ be a positive integer and let $\kappa<e$ be a non-negative integer. Choose $\lambda_j\in\mathbb{F}_{p^\upsilon}$ for $1\le j\le \ell$, where $\upsilon=\mathrm{gcd}\left(e,2\kappa\right)$. Let $\mathcal{C}$ be a $\Lambda$-MT code over $\mathbb{F}_{q}$. Define the two-sided $\kappa$-Galois dual of $\mathcal{C}$ by $\mathcal{C}^{\perp_\kappa}\cap \mathcal{C}_{\perp_\kappa}$.
\end{definition}

The condition $\lambda_j\in\mathbb{F}_{p^\upsilon}$ for all $1\le j\le \ell$ in Definition \ref{2sided} ensures that $\mathcal{C}^{\perp_\kappa}\cap\mathcal{C}_{\perp_\kappa}$ is $\sigma^\kappa \Delta$-MT. This condition was previously used in Theorems \ref{Maiin1}--\ref{Maiin6} if $\tau=1$ is chosen. Furthermore, in case $\tau=1$, Theorem \ref{Properties} shows that $\sigma^{2\kappa\tau}\left(\mathcal{C}^{\perp_\kappa}\right)$ can be replaced by $\mathcal{C}_{\perp_\kappa}$. Therefore, Theorems \ref{Maiin1}--\ref{Maiin6} have proved the following result describing the two-sided Galois dual of a MT code.

\begin{corollary}
\label{main_coroll}
Let $e$ be a positive integer and let $\kappa<e$ be a non-negative integer such that $e\mid4\kappa$. Define $\upsilon=\mathrm{gcd}\left(e,2\kappa\right)$ and let $\mathcal{C}$ be a $\left(\lambda_1,\lambda_2,\ldots,\lambda_\ell \right)$-MT code over $\mathbb{F}_{q}$, where $\lambda_j\in\mathbb{F}_{p^\upsilon}$ for $1\le j\le \ell$.  Let $\mathbf{G}$ be the reduced GPM of $\mathcal{C}$, let $\mathbf{H}$ be the reduced GPM of $\mathcal{C}^\perp$, and let $\mathbf{B}$ be the matrix that satisfies the identical equation of $\mathbf{H}$. Then
\begin{enumerate}
\item $\mathcal{C}^{\perp_\kappa}=\mathcal{C}_{\perp_\kappa}$ if and only if $\sigma^{2\kappa}\left(\mathcal{C}\right)=\mathcal{C}$ if and only if $\mathbf{G}$ is a matrix over $\mathbb{F}_{p^\upsilon}[x]$. 
\item Let $\mathbf{X}$ and $\mathbf{Y}$ be upper triangular matrices over $\mathbb{F}_{p^\upsilon}[x]$ and $\mathbb{F}_{q}[x]$, respectively, such that
\begin{enumerate}
\item $\mathbf{Y}\sigma^{e-\kappa}\left( \mathbf{H}\right)$ is a matrix over $\mathbb{F}_{p^\upsilon}[x]$, 
\item $\mathbf{X}\mathbf{Y}=\sigma^{e-\kappa}\left( \mathbf{B}\right)$, and
\item $\deg\left(\mathrm{det}\left(\mathbf{X}\right)\right)$ is maximum among all matrices that satisfy these conditions, or equivalently, $\deg\left(\mathrm{det}\left(\mathbf{Y}\right)\right)$ is minimum among all matrices that satisfy these conditions.
\end{enumerate}
Then $\mathcal{C}^{\perp_\kappa} \cap \mathcal{C}_{\perp_\kappa}$ has a GPM $\mathbf{Y}\sigma^{e-\kappa}\left( \mathbf{H}\right)$ and dimension $\deg\left(\mathrm{det}\left(\mathbf{X}\right)\right)$. Furthermore, $\mathbf{X}$ is the matrix that satisfies the identical equation of $\mathbf{Y}\sigma^{e-\kappa}\left( \mathbf{H}\right)$.
\end{enumerate}
\end{corollary}

\begin{remark}
\begin{enumerate}
\item In Corollary \ref{main_coroll}, we observed that $\mathcal{C}^{\perp_\kappa}=\mathcal{C}_{\perp_\kappa}$ if and only if $\mathbf{G}$ is a matrix over $\mathbb{F}_{p^\upsilon}[x]$. From Definition \ref{def_parity}, this is the case if and only if $\mathbf{H}$ is over $\mathbb{F}_{p^\upsilon}[x]$, hence will be $\sigma^{e-\kappa}\left( \mathbf{H}\right)$ and $\sigma^{e-\kappa}\left( \mathbf{B}\right)$ as well. In this case, the conditions given in Corollary \ref{main_coroll} are satisfied by $\mathbf{X}=\sigma^{e-\kappa}\left( \mathbf{B}\right)$ and $\mathbf{Y}=\mathbf{I}_\ell$. More precisely, $\mathcal{C}^{\perp_\kappa}=\mathcal{C}_{\perp_\kappa}$ if and only if the conditions given in Corollary \ref{main_coroll} are satisfied by an invertible $\mathbf{Y}$. On the other hand, the case of an invertible $\mathbf{X}$ is examined in Corollary \ref{direct_sum}.
\item There is a pair $\left(\mathbf{X},\mathbf{Y}\right)$ that satisfies the conditions of Corollary \ref{main_coroll}. To see this, consider the set $S=\left\{\left(\mathbf{X}_i,\mathbf{Y}_i\right)\right\}_{i\in\mathcal{I}}$ of all pairs that satisfy the first two conditions of Corollary \ref{main_coroll}. Clearly $\left(\mathbf{I}_\ell,\sigma^{e-\kappa}\left( \mathbf{B}\right)\right)\in S$. Thus $S$ is not empty. Moreover, $S$ is totally ordered by $\deg\left(\mathrm{det}\left(\mathbf{X}_i\right)\right)$. Since $0\le \deg\left(\mathrm{det}\left(\mathbf{X}_i\right)\right) \le \deg\left(\mathrm{det}\left(\mathbf{B}\right)\right)$ for every $i\in\mathcal{I}$, there is a maximal element $\left(\mathbf{X},\mathbf{Y}\right)\in S$. 
\item It is not straightforward to determine matrices $\mathbf{X}$ and $\mathbf{Y}$ that satisfy the conditions given in Corollary \ref{main_coroll}. This is because these matrices have entries in different rings, $\mathbb{F}_{p^\upsilon}[x]$ and $\mathbb{F}_{q}[x]$. We propose an auxiliary equation that may be useful in determining such matrices. If $\alpha\in\mathbb{F}_{q}$, the trace of $\alpha$, written $\mathrm{Tr}\left(\alpha\right)$, is defined by 
\begin{equation*}
\mathrm{Tr}\left(\alpha\right)=\alpha+\sigma^\upsilon\left(\alpha\right)+\sigma^{2\upsilon}\left(\alpha\right)+\cdots+\sigma^{e-\upsilon}\left(\alpha\right) \in\mathbb{F}_{p^\upsilon}.
\end{equation*}
For any $a,b\in\mathbb{F}_{p^\upsilon}$ and $\alpha,\beta\in\mathbb{F}_{q}$, we have $\mathrm{Tr}\left(a \alpha+b \beta\right)=a \mathrm{Tr}(\alpha)+ b \mathrm{Tr}(\beta)$. We can extend $\mathrm{Tr}$ to an additive group homomorphism from $\mathbb{F}_{q}[x]$ to $\mathbb{F}_{p^\upsilon}[x]$ by defining
\begin{equation*}
\mathrm{Tr}\left(a_0+a_1 x +\cdots +a_m x^m\right) = \mathrm{Tr}(a_0)+\mathrm{Tr}(a_1) x +\cdots +\mathrm{Tr}(a_m) x^m.
\end{equation*}
Similarly, for matrices over $\mathbb{F}_{q}[x]$, $\mathrm{Tr}$ defines a group homomorphism. If $\mathbf{Y}=\left[y_{i,j}\right]$ is a matrix over $\mathbb{F}_{q}[x]$, we define $\mathrm{Tr}\left(\mathbf{Y}\right)=\left[\mathrm{Tr}\left(y_{i,j}\right)\right]$. Definitely, $\mathrm{Tr}\left(\mathbf{Y}\right)$ is a matrix over $\mathbb{F}_{p^\upsilon}[x]$.

Suppose $\mathbf{X}$ and $\mathbf{Y}$ as defined in Corollary \ref{main_coroll}. For any $0\le i\le (e-\upsilon)/\upsilon$, the automorphism $\sigma^{i \upsilon}$ fixes $\mathbf{X}$ in the ring of matrices over $\mathbb{F}_{q}[x]$. Therefore,
\begin{equation}
\label{pre_aux}
\mathbf{X}\sigma^{i\upsilon}\left(\mathbf{Y}\right)=\sigma^{i\upsilon}\left(\sigma^{e-\kappa}\left( \mathbf{B}\right)\right) \quad \text{for } i=0,1,\ldots, \frac{e-\upsilon}{\upsilon}.
\end{equation}
By summing \eqref{pre_aux} over all values of $i$, we get the auxiliary equation
\begin{equation}
\label{trace}
\mathbf{X}\mathrm{Tr}\left(\mathbf{Y}\right)=\mathrm{Tr}\left(\sigma^{e-\kappa}\left( \mathbf{B}\right)\right).
\end{equation}
All matrices in \eqref{trace} are over $\mathbb{F}_{p^\upsilon}[x]$. In Example \ref{Example_one} below, we will indicate how to use \eqref{trace} to determine matrices that satisfy the conditions given in Corollary \ref{main_coroll}.
\item Some extra conditions must be taken into account when one aims to make $\mathbf{Y}\sigma^{e-\kappa}\left( \mathbf{H}\right)$ the reduced GPM of $\mathcal{C}^{\perp_\kappa}\cap\mathcal{C}_{\perp_\kappa}$. In this situation, $\mathbf{X}=\left[x_{i,j}\right]$ is the matrix that satisfies the identical equation of the reduced GPM. Then for each $1\le i\le \ell$, $x_{i,i}$ is a nonzero monic polynomial and $\deg{x_{i,j}}<\deg{x_{i,i}}$ for all $j>i$.
\end{enumerate}
\end{remark}

It is worthwhile to present a complete example illustrating the process of determining the reduced GPM of the right Galois dual, the left Galois dual, and the two-sided Galois dual of a MT code.

\begin{example}
\label{Example_one}
Let $\theta$ be a root of the irreducible polynomial $x^4+x+1 \in\mathbb{F}_2[x]$. We represent $\mathbb{F}_{16}$ as the set $\left\{a+b\theta+c\theta^2+d\theta^3\mid a,b,c,d \in\mathbb{F}_2 \right\}$. Consider the $\left(1,\theta^{10},\theta^{10}\right)$-MT code $\mathcal{C}$ over $\mathbb{F}_{16}$ of block lengths $\left(3,4,4\right)$ whose reduced GPM is
\begin{equation*}
\mathbf{G}=\begin{pmatrix}
\theta^5+\theta^{10}x+x^2 \ &\  0 \ &\  \theta^2+\theta^7x+ \theta^{12} x^2+\theta^2 x^3 \\
0\ &\ 1\ &\ 1+\theta x+\theta^{5} x^2+\theta^2 x^3\\
0\ &\ 0\ &\ \theta^{10}+x^4\\
\end{pmatrix}.
\end{equation*}
The matrix that satisfies the identical equation of $\mathbf{G}$ is
\begin{equation*}
\mathbf{A}=\begin{pmatrix}
\theta^{10}+x\ &\ 0\ &\  \theta^2\\
0\ &\ \theta^{10}+x^4\ &\  1+\theta x+\theta^5 x^2+\theta^2 x^3\\
0\ &\ 0\ &\ 1\\
\end{pmatrix}.
\end{equation*}
The dimension of $\mathcal{C}$ is $k=5$ and its minimum distance is $d_{\mathrm{min}}=5$. It follows from Theorem \ref{MT_dual} that $\mathcal{C}^\perp$ is $\left(1,\theta^5,\theta^5\right)$-MT of block lengths $\left(3,4,4 \right)$ and dimension $6$. Theorem \ref{MT_Parity_H} provides a GPM for $\mathcal{C}^\perp$ whose Hermite normal form is
\begin{equation*}
\mathbf{H}=\begin{pmatrix}
1\ &\ \theta^9\ &\  \theta^9+x+\theta x^2 +\theta^9 x^3\\
0\ &\ \theta^5+x\ &\  \theta^{12} x +\theta^4 x^2+\theta^{13} x^3\\
0\ &\ 0\ &\ \theta^5+x^4\\
\end{pmatrix}.
\end{equation*}
The matrix that satisfies the identical equation of $\mathbf{H}$ is
\begin{equation*}
\mathbf{B}=\begin{pmatrix}
1+x^3\ &\ \theta^4+\theta^{14} x+\theta^9 x^2\ &\  \theta^4+\theta^{14} x+\theta^9 x^2\\
0\ &\ 1+\theta^{10} x+\theta^5 x^2+x^3\ &\  \theta^{7} x +\theta^{13} x^2\\
0\ &\ 0\ &\ 1\\
\end{pmatrix}.
\end{equation*}

We consider the $3$-Galois inner product on $\mathbb{F}_{16}^{11}$. Observe that $\lambda_1=1$ and $\lambda_2=\lambda_3=\theta^{10}$ are elements of $\mathbb{F}_{p^\upsilon}=\mathbb{F}_4=\{0,1,\theta^5,\theta^{10}\}$, where $\upsilon=\mathrm{gcd}\left(e,2\kappa\right)=2$. Theorem \ref{right_dual} ensures that $\mathcal{C}^{\perp_3}$ is $\left(1,\theta^{10},\theta^{10} \right)$-MT of block lengths $\left(3,4,4 \right)$ and dimension $6$. Moreover, the reduced GPM of $\mathcal{C}^{\perp_3}$ is 
\begin{equation*}
\sigma^{e-3}\left(\mathbf{H}\right)=\sigma\left(\mathbf{H}\right)=\begin{pmatrix}
1\ &\ \theta^3\ &\  \theta^3+x+\theta^2 x^2 +\theta^3 x^3\\
0\ &\ \theta^{10}+x\ &\  \theta^9 x +\theta^8 x^2+\theta^{11} x^3\\
0\ &\ 0\ &\ \theta^{10}+x^4\\
\end{pmatrix},
\end{equation*}
and the matrix that satisfies its identical equation is  
\begin{equation*}
\sigma^{e-3}\left(\mathbf{B}\right)=\sigma\left(\mathbf{B}\right)=\begin{pmatrix}
1+x^3\ &\ \theta^8+\theta^{13} x+\theta^3 x^2\ &\  \theta^8+\theta^{13} x+\theta^3 x^2\\
0\ &\ 1+\theta^5 x+\theta^{10} x^2+x^3\ &\  \theta^{14} x +\theta^{11} x^2\\
0\ &\ 0\ &\ 1\\
\end{pmatrix}.
\end{equation*}
However, Theorem \ref{left_dual} ensures that $\mathcal{C}_{\perp_3}$ is $\left(1,\theta^{10},\theta^{10} \right)$-MT of block lengths $\left(3,4,4 \right)$ and dimension $6$. Moreover, the reduced GPM of $\mathcal{C}_{\perp_3}$ is
\begin{equation*}
\sigma^{3}\left(\mathbf{H}\right)=\begin{pmatrix}
1\ &\ \theta^{12}\ &\  \theta^{12}+x+\theta^8 x^2 +\theta^{12} x^3\\
0\ &\ \theta^{10}+x\ &\  \theta^{6} x +\theta^2 x^2+\theta^{14} x^3\\
0\ &\ 0\ &\ \theta^{10}+x^4\\
\end{pmatrix},
\end{equation*}
and the matrix that satisfies its identical equation is 
\begin{equation*}
\sigma^{3}\left(\mathbf{B}\right)=\begin{pmatrix}
1+x^3\ &\ \theta^2+\theta^{7} x+\theta^{12} x^2\ &\  \theta^2+\theta^{7} x+\theta^{12} x^2\\
0\ &\ 1+\theta^{5} x+\theta^{10} x^2+x^3\ &\  \theta^{11} x +\theta^{14} x^2\\
0\ &\ 0\ &\ 1\\
\end{pmatrix}.
\end{equation*}

It remains to compute the reduced GPM of the two-sided $3$-Galois dual of $\mathcal{C}$ with the aid of Corollary \ref{main_coroll}. Since not all entries of $\mathbf{G}$ are elements of $\mathbb{F}_{4}[x]$, we conclude that $\mathcal{C}^{\perp_3}\ne \mathcal{C}_{\perp_3}$. Suppose $\mathbf{X}=\left[x_{i,j}\right]$ and $\mathbf{Y}=\left[y_{i,j}\right]$ are upper triangular matrices that satisfy the conditions given in Corollary \ref{main_coroll}. That is, $\mathrm{det}\left(\mathbf{X}\right)$ has the maximum possible degree, the product $\mathbf{Y}\sigma\left(\mathbf{H}\right)$ yields a matrix over $\mathbb{F}_4[x]$, and
\begin{equation}
\label{In_Ex_Eq_1}
\begin{split}
\mathbf{X}\mathbf{Y}&=\begin{pmatrix}
x_{11}\ &\  x_{12} \ &\  x_{13}\\
0\ &\  x_{22} \ &\  x_{23}\\
0\ &\ 0\ &\ x_{33}\\
\end{pmatrix} \begin{pmatrix}
y_{11}\ &\  y_{12} \ &\  y_{13}\\
0\ &\  y_{22} \ &\  y_{23}\\
0\ &\ 0\ &\ y_{33}\\
\end{pmatrix}\\
&=\sigma\left(\mathbf{B}\right)=\begin{pmatrix}
1+x^3\ &\ \theta^3 (1+x)(\theta^5 +x)\ &\  \theta^3 (1+x)(\theta^5 +x)\\
0\ &\ (\theta^{10} +x)^3\ &\  \theta^{11} x(\theta^3+x)\\
0\ &\ 0\ &\ 1\\
\end{pmatrix}
\end{split}
\end{equation}
where $x_{i,j}\in\mathbb{F}_{4}[x]$ and $y_{i,j}\in\mathbb{F}_{16}[x]$ for $1\le i,j \le 3$. In addition, by using the trace map $\mathrm{Tr}:\alpha \mapsto \alpha+\alpha^{4}$ for any $\alpha\in\mathbb{F}_{16}$, we utilize the auxiliary equation \eqref{trace}:
\begin{equation}
\label{In_Ex_Eq_0}
\begin{split}
\mathbf{X}\mathrm{Tr}\left(\mathbf{Y}\right)&=\begin{pmatrix}
x_{11}\ &\  x_{12} \ &\  x_{13}\\
0\ &\  x_{22} \ &\  x_{23}\\
0\ &\ 0\ &\ x_{33}\\
\end{pmatrix} \begin{pmatrix}
\mathrm{Tr}\left(y_{11}\right)\ &\  \mathrm{Tr}\left(y_{12}\right) \ &\  \mathrm{Tr}\left(y_{13}\right)\\
0\ &\  \mathrm{Tr}\left(y_{22}\right) \ &\  \mathrm{Tr}\left(y_{23}\right)\\
0\ &\ 0\ &\ \mathrm{Tr}\left(y_{33}\right)\\
\end{pmatrix}\\
&=\mathrm{Tr}\left(\sigma^{e-\kappa}\left( \mathbf{B}\right)\right)=\begin{pmatrix}
0\ &\ \theta^{10}(1+x)(\theta^{5} +x)\ &\  \theta^{10}(1+x)(\theta^{5} +x)\\
0\ &\ 0\ &\  \theta^{10} x (1+x)\\
0\ &\ 0\ &\ 0\\
\end{pmatrix}.
\end{split}
\end{equation}

From \eqref{In_Ex_Eq_0}, $\mathrm{Tr}\left(y_{11}\right)=\mathrm{Tr}\left(y_{22}\right)=\mathrm{Tr}\left(y_{33}\right)=0$ since $x_{11},x_{22},x_{33}\ne 0$. From \eqref{In_Ex_Eq_1}, $x_{33}=y_{33}=1$. Observe that $x_{22}$ divides $(\theta^{10}+x)^3$ by \eqref{In_Ex_Eq_1} and divides $\theta^{10}x(1+x)$ by \eqref{In_Ex_Eq_0}, and thus $x_{22}=1$. Since $\mathbf{Y}\sigma\left(\mathbf{H}\right)$ is assumed to be in the reduced form, $\deg{x_{23}}<\deg{x_{22}}=0$ and hence $x_{23}=0$. From \eqref{In_Ex_Eq_1}, $y_{22}=(\theta^{10}+x)^3$ and $y_{23}=\theta^{11} x(\theta^3+x)$. In the same way, observe that $x_{11}$ divides $1+x^3=(1+x)(\theta^5+x)(\theta^{10}+x)$ by \eqref{In_Ex_Eq_1} and divides $\theta^{10}(1+x)(\theta^{5} +x)$ by \eqref{In_Ex_Eq_0}. Thus, the maximum possible degree of $x_{11}$ is obtained by taking $x_{11}=(1+x)(\theta^{5} +x)$, and then $y_{11}=\theta^{10}+x$.  From \eqref{In_Ex_Eq_1}, \[(1+x)(\theta^{5} +x)y_{12}+(\theta^{10}+x)^3 x_{12}=\theta^3 (1+x)(\theta^{5} +x).\]
Then $(1+x)(\theta^{5} +x)$ divides $x_{12}$. Thus, $x_{12}= 0$ and $y_{12}=\theta^3$ since otherwise $\deg{x_{12}}\ge 2=\deg{x_{11}}$. Again from \eqref{In_Ex_Eq_1}, \[(1+x)(\theta^{5} +x)y_{13}+ x_{13}=\theta^3 (1+x)(\theta^5 +x).\]
Then $(1+x)(\theta^{5} +x)$ divides $x_{13}$. Thus, $x_{13}=0$ and $y_{13}=\theta^3$ since otherwise $\deg{x_{13}}\ge 2=\deg{x_{11}}$. It is easy to check that this solution makes the entries of $\mathbf{Y}\sigma\left(\mathbf{H}\right)$ elements of $\mathbb{F}_4[x]$. So far, the conditions given in Corollary \ref{main_coroll} have been satisfied for
\begin{equation*}
\mathbf{X}=\begin{pmatrix}
(1+x)(\theta^{5} +x)\ &\  0 \ &\  0\\
0\ &\  1 \ &\  0\\
0\ &\ 0\ &\ 1\\
\end{pmatrix}\  \text{ and }\ 
\mathbf{Y}= \begin{pmatrix}
\theta^{10}+x\ &\  \theta^3 \ &\  \theta^3\\
0\ &\  (\theta^{10}+x)^3 \ &\  \theta^{11} x(\theta^3+x)\\
0\ &\ 0\ &\ 1\\
\end{pmatrix}.
\end{equation*}
Therefore, the reduced GPM of the two-sided $3$-Galois dual of $\mathcal{C}$ is 
\begin{equation*}
\begin{split}
\mathbf{Y}\sigma\left(\mathbf{H}\right)&=
\begin{pmatrix}
\theta^{10}+x &  \theta^3  &  \theta^3\\
0 &  (\theta^{10}+x)^3  &  \theta^{11} x(\theta^3+x)\\
0 & 0 & 1\\
\end{pmatrix}\begin{pmatrix}
1 & \theta^3 &  \theta^3+x+\theta^2 x^2 +\theta^3 x^3\\
0 & \theta^{10}+x &  \theta^9 x +\theta^8 x^2+\theta^{11} x^3\\
0 & 0 & \theta^{10}+x^4\\
\end{pmatrix}\\
&=\begin{pmatrix}
\theta^{10}+x\ &\ 0\ &\  0\\
0\ &\ \theta^{10}+x^4\ &\  0\\
0\ &\ 0\ &\ \theta^{10}+x^4\\
\end{pmatrix}.
\end{split}
\end{equation*}
The dimension of $\mathcal{C}^{\perp_3}\cap \mathcal{C}_{\perp_3}$ is $\deg\left(\mathrm{det}\left(\mathbf{X}\right)\right)=2$. The diagonalizability of the reduced GPM of $\mathcal{C}^{\perp_3}\cap \mathcal{C}_{\perp_3}$ indicates that $\mathcal{C}^{\perp_3}\cap \mathcal{C}_{\perp_3}$ is the direct sum of constacyclic codes. More precisely, $\mathcal{C}^{\perp_3}\cap \mathcal{C}_{\perp_3}=\mathcal{C}_1\oplus \mathbf{0}\oplus \mathbf{0}$, where $\mathcal{C}_1$ is the cyclic code of length $3$ over $\mathbb{F}_{16}$ with generator polynomial $\theta^{10}+x$ and $\mathbf{0}$ is the zero code of length $4$.
\end{example}

We conclude this section with an application of Corollary \ref{main_coroll} concerning when the right and the left Galois duals trivially intersect. Obviously, $\mathcal{C}^{\perp_\kappa}\cap \mathcal{C}_{\perp_\kappa}=\{\mathbf{0}\}$ if and only if the first two conditions of Corollary \ref{main_coroll} can only be satisfied by an invertible $\mathbf{X}$, hence zero is the maximum of $\deg\left(\mathrm{det}\left(\mathbf{X}\right)\right)$. We remark that these two conditions are always satisfied by an invertible $\mathbf{X}$, for instance $\mathbf{X}=\mathbf{I}_\ell$ and $\mathbf{Y}=\sigma^{e-\kappa}\left( \mathbf{B}\right)$. However, if these two conditions are never satisfied except for an invertible $\mathbf{X}$, then $\mathcal{C}^{\perp_\kappa}\cap \mathcal{C}_{\perp_\kappa}=\{\mathbf{0}\}$. The following is a particular case that requires the code dimension to be half the code length.

\begin{corollary}
\label{direct_sum}
Let $e$ be a positive integer and let $\kappa<e$ be a non-negative integer such that $e\mid4\kappa$. Define $\upsilon=\mathrm{gcd}\left(e,2\kappa\right)$ and let $\mathcal{C}$ be a $\left(\lambda_1,\lambda_2,\ldots,\lambda_\ell \right)$-MT code over $\mathbb{F}_{q}$ of length $n$ and dimension $n/2$, where $\lambda_j\in\mathbb{F}_{p^\upsilon}$ for $1\le j\le \ell$.  Let $\mathbf{H}$ be the reduced GPM of $\mathcal{C}^\perp$ and let $\mathbf{B}$ be the matrix that satisfies the identical equation of $\mathbf{H}$. Let $\mathbf{X}$ and $\mathbf{Y}$ be upper triangular matrices over $\mathbb{F}_{p^\upsilon}[x]$ and $\mathbb{F}_{q}[x]$, respectively, such that
\begin{enumerate}
\item $\mathbf{Y}\sigma^{e-\kappa}\left( \mathbf{H}\right)$ is a matrix over $\mathbb{F}_{p^\upsilon}[x]$, and
\item $\mathbf{X}\mathbf{Y}=\sigma^{e-\kappa}\left( \mathbf{B}\right)$.
\end{enumerate}
Then $\mathbb{F}_{q}^n=\mathcal{C}^{\perp_\kappa}\oplus \mathcal{C}_{\perp_\kappa}$ if and only if the above conditions can only be satisfied by an invertible $\mathbf{X}$.
\end{corollary}
\begin{proof}
From Corollary \ref{main_coroll}, there exist $\mathbf{X}$ and $\mathbf{Y}$ such that $\deg{\left(\mathrm{det}\left(\mathbf{X}\right)\right)}$ is the dimension of $\mathcal{C}^{\perp_\kappa}\cap \mathcal{C}_{\perp_\kappa}$. If the given conditions can only be satisfied by an invertible $\mathbf{X}$, then $\mathcal{C}^{\perp_\kappa}\cap \mathcal{C}_{\perp_\kappa}$ has dimension zero. Therefore, $\mathcal{C}^{\perp_\kappa}\cap \mathcal{C}_{\perp_\kappa}=\{\mathbf{0}\}$ and $\mathbb{F}_{q}^n=\mathcal{C}^{\perp_\kappa}\oplus \mathcal{C}_{\perp_\kappa}$ because both $\mathcal{C}^{\perp_\kappa}$ and $\mathcal{C}_{\perp_\kappa}$ have dimension $n/2$.

Conversely, suppose that $\mathbf{X}$ and $\mathbf{Y}$ satisfy the given conditions. From Theorem \ref{Maiin5}, there exists a subcode $\mathcal{S}$ of $\mathcal{C}^{\perp_\kappa}\cap \mathcal{C}_{\perp_\kappa}$ such that $\mathcal{S}$ has dimension $\deg\left(\mathrm{det}\left(\mathbf{X} \right)\right)$. Assume that $\mathbb{F}_{q}^n=\mathcal{C}^{\perp_\kappa}\oplus \mathcal{C}_{\perp_\kappa}$. Then $\mathcal{S}\subseteq\mathcal{C}^{\perp_\kappa}\cap \mathcal{C}_{\perp_\kappa}=\{\mathbf{0}\}$, and hence $\deg{\left(\mathrm{det}\left(\mathbf{X}\right)\right)}=0$. That is, $\mathbf{X}$ is invertible.
\end{proof}

\begin{example}
Consider $p=3$ and $e=4$. We can identify the elements of $\mathbb{F}_{81}$ with polynomials of the form $a_0+a_1 \theta+a_2 \theta^2+a_3 \theta^3$, where $a_i\in\mathbb{F}_3$ for $0\le i\le 3$ and $\theta$ is a root of the irreducible polynomial $x^4+x+2\in\mathbb{F}_3[x]$. Let $\mathcal{C}$ be the $\left(\theta^{50},\theta^{20}\right)$-MT code over $\mathbb{F}_{81}$ of block lengths $(4,8)$ and reduced GPM
\begin{equation*}
\mathbf{G}=\begin{pmatrix}
1 \ &\  2+\theta^5 x^2+ \theta^{10} x^4 \\
0\ &\ \theta^{55}+\theta^{10} x^2+\theta^{45} x^4+ x^6\\
\end{pmatrix}.
\end{equation*}
The matrix that satisfies the identical equation of $\mathbf{G}$ is 
\begin{equation*}
\mathbf{A}=\begin{pmatrix}
x^4-\theta^{50} \ &\  \theta^{25}+2 x^2 \\
0\ &\ \theta^{5}+ x^2\\
\end{pmatrix},
\end{equation*} 
from which we conclude that $\mathcal{C}$ has a dimension equal to half the code length. By Theorem \ref{MT_dual} and Theorem \ref{MT_Parity_H}, the Euclidean dual $\mathcal{C}^\perp$ of $\mathcal{C}$ is $\left(\theta^{30},\theta^{60} \right)$-MT with a GPM whose reduced form is
\begin{equation*}
\mathbf{H}=\begin{pmatrix}
\theta^{15}+x^2 \ &\  \theta^{75}+  x^2 \\
0\ &\ \theta^{50}+ \theta^5 x^2 +x^4\\
\end{pmatrix},
\end{equation*}
and the matrix that satisfies the identical equation of $\mathbf{H}$ is 
\begin{equation*}
\mathbf{B}=\begin{pmatrix}
\theta^{55}+x^2 \ &\  2 \\
0\ &\ \theta^{50}+ \theta^{45} x^2 +x^4\\
\end{pmatrix}.
\end{equation*}
Take $\kappa=1$ and notice that $\lambda_1, \lambda_2\in\mathbb{F}_{3^\upsilon}=\mathbb{F}_{9}$, where $\upsilon=\mathrm{gcd}\left(e,2\kappa\right)=2$. In fact, $\mathbb{F}_{9}=\left\{0,1,\theta^{10},\theta^{20},\theta^{30},\theta^{40},\theta^{50},\theta^{60},\theta^{70}\right\}$. Suppose that $\mathbf{X}$ and $\mathbf{Y}$ satisfy the conditions given in Corollary \ref{direct_sum}. We show that $\mathbf{X}$ is invertible. With this hypothesis, the matrix obtained from
\begin{equation}
\label{In_Ex2_1}
\mathbf{Y}\sigma^3\left(\mathbf{H}\right)=
\begin{pmatrix}
y_{11}\ &\  y_{12}\\
0\ &\  y_{22} \\
\end{pmatrix}\begin{pmatrix}
\theta^{5}+x^2 \ &\  \theta^{25}+  x^2 \\
0\ &\ (\theta^{25}+  x^2)(\theta^{45}+  x^2)\\
\end{pmatrix}
\end{equation}
has entries in $\mathbb{F}_9[x]$ and
\begin{equation}
\label{In_Ex2_2}
\mathbf{X}\mathbf{Y}=\begin{pmatrix}
x_{11}\ &\  x_{12}\\
0\ &\  x_{22}\\
\end{pmatrix} \begin{pmatrix}
y_{11}\ &\  y_{12}\\
0\ &\  y_{22}\\
\end{pmatrix}
=\begin{pmatrix}
\theta^{45}+x^2 \ &\  2 \\
0\ &\ (\theta^{5}+  x^2)(\theta^{65}+  x^2)\\
\end{pmatrix}.
\end{equation}
From \eqref{In_Ex2_1}, $\left(\theta^{45}+x^2\right)$ divides $y_{11}$ and $(\theta^{5}+  x^2)(\theta^{65}+  x^2)$ divides $y_{11}$ because it is required to make $y_{11}\left(\theta^{5}+x^2\right)$ and $y_{22}\left(\theta^{25}+  x^2\right)\left(\theta^{45}+  x^2\right)$ elements in $\mathbb{F}_{9}[x]$. It follows from \eqref{In_Ex2_2} that $x_{11}=x_{22}=1$ and, hence, $\mathbf{X}$ is invertible. By Corollary \ref{direct_sum}, $\mathbb{F}_{81}^{12}$ can be written as the direct sum of the right and left $1$-Galois duals of $\mathcal{C}$, that is, $\mathbb{F}_{81}^{12}=\mathcal{C}^{\perp_1}\oplus\mathcal{C}_{\perp_1}$.
\end{example}

\section{Conclusion}
\label{conclusion}
MT codes constitute a comprehensive class of linear codes. It contains cyclic, constacyclic, QC, QT, and GQC codes as subclasses. We introduce the properties of MT codes and present them as free modules over $\mathbb{F}_q[x]$. This algebraic structure identifies the MT code by GPM which satisfies an identical equation. We use this identical equation to provide a GPM formula for the Euclidean dual. Then we use the Galois inner product to generalize our findings on the Euclidean dual. It makes sense, then, to differentiate between the right and left Galois duals of a linear code. We demonstrate a necessary and sufficient condition for these two duals to be identical. We discuss the right and left Galois duals of linear codes in general (and MT codes in particular) as well as their related properties. We additionally define the two-sided Galois dual, which has not been previously defined in any former study. We characterize a GPM for the two-sided Galois dual of a MT code and illustrate its construction with a detailed example. Finally, we give an equivalent condition under which the vector space $\mathbb{F}_{q}^n$ can be written as a direct sum of the right and left Galois duals of a MT code.



\begin{thebibliography}{10}
\providecommand{\url}[1]{{#1}}
\providecommand{\urlprefix}{URL }
\expandafter\ifx\csname urlstyle\endcsname\relax
  \providecommand{\doi}[1]{DOI~\discretionary{}{}{}#1}\else
  \providecommand{\doi}{DOI~\discretionary{}{}{}\begingroup
  \urlstyle{rm}\Url}\fi

\bibitem{Aydin2017}
Aydin, N., Halilovi{\'{c}}, A.: A generalization of quasi-twisted codes:
  Multi-twisted codes.
\newblock Finite Fields and Their Applications \textbf{45}, 96--106 (2017)

\bibitem{Bakshi2012}
Bakshi, G.K., Raka, M.: A class of constacyclic codes over a finite field.
\newblock Finite Fields and Their Applications \textbf{18}(2), 362--377 (2012)

\bibitem{Barbier2012}
Barbier, M., Chabot, C., Quintin, G.: On quasi-cyclic codes as a generalization
  of cyclic codes.
\newblock Finite Fields and Their Applications \textbf{18}(5), 904--919 (2012)

\bibitem{Cayrel2010}
Cayrel, P.L., Chabot, C., Necer, A.: Quasi-cyclic codes as codes over rings of
  matrices.
\newblock Finite Fields and Their Applications \textbf{16}(2), 100--115 (2010)

\bibitem{chauhan2021}
Chauhan, V.: Multi-twisted codes over finite fields and their generalizations.
\newblock Ph.D. thesis, Indraprastha Institute of Information Technology (2021)

\bibitem{Sharma}
Chauhan, V., Sharma, A.: A generalization of multi-twisted codes over finite
  fields, their galois duals and type {II} codes.
\newblock Journal of Applied Mathematics and Computing \textbf{68}(2),
  1413--1447 (2021)

\bibitem{Chen2012}
Chen, B., Fan, Y., Lin, L., Liu, H.: Constacyclic codes over finite fields.
\newblock Finite Fields and Their Applications \textbf{18}(6), 1217--1231
  (2012)

\bibitem{Cohen1996}
Cohen, H.: Hermite and smith normal form algorithms over dedekind domains.
\newblock Mathematics of Computation \textbf{65}(216), 1681--1699 (1996)

\bibitem{Esmaeili2009}
Esmaeili, M., Yari, S.: Generalized quasi-cyclic codes: structural properties
  and code construction.
\newblock Applicable Algebra in Engineering, Communication and Computing
  \textbf{20}(2), 159--173 (2009)

\bibitem{Fan2016}
Fan, Y., Zhang, L.: Galois self-dual constacyclic codes.
\newblock Designs, Codes and Cryptography \textbf{84}(3), 473--492 (2016)

\bibitem{Gao2014}
Gao, J., Fu, F.W.: Note on quasi-twisted codes and an application.
\newblock Journal of Applied Mathematics and Computing \textbf{47}(1-2),
  487--506 (2014)

\bibitem{Gneri2017}
Güneri, C., Özbudak, F., Özkaya, B., Saçıkara, E., Sepasdar, Z., Solé,
  P.: Structure and performance of generalized quasi-cyclic codes.
\newblock Finite Fields and Their Applications \textbf{47}, 183--202 (2017)

\bibitem{Jia2012quasi}
Jia, Y.: On quasi-twisted codes over finite fields.
\newblock Finite Fields and Their Applications \textbf{18}(2), 237--257 (2012)

\bibitem{Lally2001}
Lally, K., Fitzpatrick, P.: Algebraic structure of quasicyclic codes.
\newblock Discrete Applied Mathematics \textbf{111}(1-2), 157--175 (2001)

\bibitem{SanLing2001}
Ling, S., Solé, P.: On the algebraic structure of quasi-cyclic codes .i.
  finite fields.
\newblock {IEEE} Transactions on Information Theory \textbf{47}(7), 2751--2760
  (2001)

\bibitem{Hongwei}
Liu, H., Pan, X.: Galois hulls of linear codes over finite fields.
\newblock Designs, Codes and Cryptography \textbf{88}(2), 241--255 (2019)

\bibitem{Matsui2015}
Matsui, H.: On generator and parity-check polynomial matrices of generalized
  quasi-cyclic codes.
\newblock Finite Fields and Their Applications \textbf{34}, 280--304 (2015)

\bibitem{Sharma2018}
Sharma, A., Chauhan, V., Singh, H.: Multi-twisted codes over finite fields and
  their dual codes.
\newblock Finite Fields and Their Applications \textbf{51}, 270--297 (2018)

\bibitem{Siap05thestructure}
Siap, I., Kulhan, N.: The structure of generalized quasi-cyclic codes.
\newblock Applied Mathematics E-Notes \textbf{5}, 24--30 (2005)

\end{thebibliography}
\end{document}